\documentclass[final,1p,times,authoryear]{elsarticle}
\usepackage{amsmath}
\usepackage{amsthm}
\usepackage{amssymb}
\usepackage{amsfonts}

\usepackage{mathrsfs}
\usepackage{amscd}
\usepackage{CJK}
\usepackage{color}
\usepackage{bbm}
\usepackage{graphicx}
\usepackage{setspace}
\usepackage{algorithm}
\usepackage{algorithmic}
\newtheorem{theorem}{Theorem}[section]
\newtheorem{prop}[theorem]{Proposition}

\newtheorem{lemma}[theorem]{Lemma}
\newtheorem{define}[theorem]{Definition}
\newtheorem{remark}[theorem]{Remark}

\def\no{\nonumber}

\def\F{{\mathcal{F}}}
\def\Q{{\mathbb{Q}}}
\def\dtrdeg{\hbox{$\sigma$\rm{.tr.deg}}}
\def\Y{{\mathbb{Y}}}
\def\bu{{\mathbf{u}}}
\def\P{{\mathbb P}}
\def\norm{\hbox{\rm{N}}}

\floatstyle{ruled}
\newfloat{algorithm}{t}{loa}
\floatname{algorithm}{Algorithm}
\newcommand{\Inp}[1]
  {\noindent\begin{tabular}{@{}p{1.8cm}@{}p{13.2cm}@{}}
   {\bf Input: }&#1 \end{tabular}}
\newcommand{\Outp}[1]
  {\noindent\begin{tabular}{@{}p{1.8cm}@{}p{13.2cm}@{}}
   {\bf Output: }&#1 \end{tabular}}

\DeclareFontFamily{U}{fsy}{} \DeclareFontShape{U}{fsy}{m}{n}{<->s*[.
9]psyr}{} \DeclareSymbolFont{der@m}{U}{fsy}{m}{n}
\DeclareMathSymbol{\diff}{\mathord}{der@m}{182}

\newcommand{\SPC}{\hspace*{15pt}}

\floatstyle{ruled}
\newfloat{algorithm}{t}{loa}
\floatname{algorithm}{Algorithm}

\def\X{{\mathbb{X}}}
\def\Y{{\mathbb{Y}}}

\def\C{{\mathcal C}}

\def\P{{\mathbb{P}}}
\def\Q{{\mathbb{Q}}}
\def\BF{{\mathbb{F}}}

\def\bu{{\mathbf{u}}}

\def\bc{{\mathbf{c}}}

\def\max{\hbox{\rm{max}}}

\def\Res{\hbox{\rm Res}}

\def\deg{\hbox{\rm{deg}}}

\def\ord{\hbox{\rm{ord}}}

\def\ord{\hbox{\rm{ord}}}
\def\rk{\hbox{\rm{rk}}}
\def\Jac{\hbox{\rm{Jac}}}

\def\I{\mathcal{I}}

\def\SR{{\mathbf{R}}}

\def\Q{{\mathbb Q}}

\def\F{{\mathcal{F}}}

\def\TT{{\mathbb{T}}}
\def\JJ{{\mathbb{J}}}

\def\dtrdeg{\hbox{\rm{d.tr.deg}}}

\def\and{\cap}

\newcounter{bean}
\def\bl{\begin{list}{Step \arabic{bean}}{\usecounter{bean}}\labelwidth=34pt}
\def\el{\end{list}}

\def\deg{{\rm deg}}

\def\normalization1{{\rm normalization1}}
\def\normalization{{\rm normalization}}

\def\irrfactor1{{\rm irrfactor1}}
\def\irrfactor{{\rm irrfactor}}

\def\card{\hbox{\rm{card}}}

\def\rank{{\rm rank}}

\def\X{{\mathbb{X}}}
\def\Y{{\mathbb{Y}}}

\def\C{{\mathcal C}}

\def\SC{{\mathcal{S}}}

\def\P{{\mathbb P}}
\def\Q{{\mathbb Q}}

\def\bu{{\mathbf{u}}}

\def\bc{{\mathbf{c}}}

\def\max{\hbox{\rm{max}}}

\def\deg{\hbox{\rm{deg}}}

\def\ord{\hbox{\rm{ord}}}

\def\ord{\hbox{\rm{ord}}}
\def\rk{\hbox{\rm{rk}}}

\def\norm{\hbox{\rm{N}}}
\def\I{\mathcal{I}}

\def\BF{{\mathbb{F}}}

\def\Q{{\mathbb Q}}

\def\F{{\mathcal {F}}}

\def\dtrdeg{\hbox{$\sigma$\rm{.tr.deg}}}
\def\Jac{\hbox{\rm{Jac}}}
\def\and{\cap}

\def\bl{\begin{list}{Step \arabic{bean}}{\usecounter{bean}}\labelwidth=34pt}
\def\el{\end{list}}

\def\deg{{\rm deg}}

\def\normalization1{{\rm normalization1}}
\def\normalization{{\rm normalization}}

\def\irrfactor1{{\rm irrfactor1}}
\def\irrfactor{{\rm irrfactor}}

\def\Res{\hbox{\rm Res}}

\def\rank{{\rm rk}}

\def\SR{{\mathbf{R}}}

\def\TT{{\mathbb{T}}}
\def\JJ{{\mathbb{J}}}
\def\card{\rm card}
\def\c{\mathrm{c}}

\begin{document}

\begin{frontmatter}

\title{New bounds and an efficient algorithm for sparse difference resultants\footnote{This paper is partially supported by the National Natural Science Foundation of China (Nos.11688101 and 11671014) and Beijing Natural Science Foundation (No. Z190004).}}

\author{Chun-Ming Yuan}
\address{KLMM, Academy of Mathematics and Systems Sciences, CAS, 100190, Beijing, China}
 \address{School of Mathematical Sciences, University of  Chinese Academy of Sciences,100049, Beijing, China}
\ead{cmyuan@mmrc.iss.ac.cn}
%\ead[url]{URL 1}

\author{Zhi-Yong Zhang* }
\address{College of Science, Minzu University of China,
Beijing 100081, China}
\ead{zzy@muc.edu.cn (Corresponding author)}
%\ead[url]{URL 2}

\begin{abstract}
 The sparse difference resultant introduced in \citep{gao-2015} is a basic concept in difference elimination theory.
 In this paper,
we show that the sparse difference resultant of a generic Laurent transformally essential system
can be computed via the sparse resultant of a simple algebraic system arising from the difference system.
Moreover, new order bounds of sparse difference resultant are found.
Then we propose an efficient algorithm to compute sparse difference resultant which is the quotient of two determinants
whose elements are the coefficients of the polynomials in the algebraic system.
 The complexity of the algorithm is analyzed and experimental results show the efficiency of the algorithm.
\end{abstract}

\end{frontmatter}
\section{Introduction}\label{SECT:intro}
It is well-known that the resultant, as a basic concept in algebraic geometry and a powerful tool in elimination theory, gives conditions for an over-determined system of polynomial equations to have common solutions \citep{cox-2004}.  Since most polynomials are sparse as they only contain certain fixed monomials, Gelfand, Kapranov, Sturmfels, and Zelevinsky introduced the concept of sparse resultant \citep{gelfand,sturmfels}. %Later a number of algorithms are proposed to compute sparse resultant. In particular,
Canny and Emiris showed that the sparse resultant is a factor of the determinant of a Macaulay style matrix and gave efficient
algorithms to compute the sparse resultant based on this matrix representation \citep{emiris,emiris-2012a,emiris-2012b}.
D'Andrea further showed that the sparse resultant is the quotient of two determinants where the denominator is a minor of the
numerator \citep{an-2011}.

With the resultant and sparse resultant theories becoming more mature, it is a natural idea to extend the algebraic results to differential and
difference cases due to their broad applications. However, such results in differential and difference cases are not complete parallel with
algebraic case. For the ordinary differential case, differential resultants and sparse differential resultants are studied successively \citep{gao-2015-a,rueda-2010,yang-2011}. %by many researchers,
%we will not state them in detail and refer to
%and references therein.
% In this paper, we concentrate on the ordinary
%difference case, that is, the sparse difference resultant.}
For the ordinary difference case, Li et al. introduced the concept of sparse difference resultant for a Laurent
transformally essential system consisting of $n+1$ Laurent difference polynomials in $n$ difference variables and its basic
properties are proved \citep{gao-2015}.
%
%the difference resultant give a necessary condition that
%the original system has common zero in Laurent case.
Based on the degree and the order bounds, they proposed a single exponential algorithm in terms of the number of variables, the Jacobi number, and the size of the Laurent transformally essential system, which is essentially to search for the sparse difference resultant with the given order and degree bound.

{\color{black} Sparse} difference resultant arises in many areas of pure and applied mathematics and has potential applications beside the symbolic computation community. For example, the vanishing of the sparse difference resultant gives a necessary condition for the corresponding difference polynomial system to have non-zero solutions \citep{gao-2015}.  {\color{black}Thus it has great potentials to solve some important problems in difference algebra}, such as elimination of difference indeterminates, solving
difference equations in difference fields and so on \citep{cohn,ovc-2018,hru-2007}. {Some natural phenomena in real world are described by difference equations \citep{gao-2009,ovc-2018,roeger-2004,hensen-2007,ekhad-2014}. } However, since the case of difference polynomial system is traditionally significantly harder, there are only few efficient algorithms and techniques about difference resultants so far. Up to now using mature techniques of the algebraic case to deal with differential or difference cases is an effective way \citep{yang-2011,gao-2015}. In the difference case one need to convert difference polynomials into an algebraic polynomial system, then to compute the classical sparse resultant, and finally convert the results back in the
world of difference polynomials. One of the key observations of the conversion is the effective order of sparse difference resultant. Thus the order bound of sparse difference resultant is particularly important and has a direct impact on the complexity of the associated algorithm.

In this paper, we further explore efficient algorithms to find the sparse difference resultant of a
given difference polynomial system {using the difference structure and difference specialization technique}.
We show that the sparse difference resultant of a Laurent transformally
essential system consisting of $n+1$ Laurent difference polynomials in $n$ difference variables
is the same as the one of a simple system consisting of $m+1$ polynomials in $m$ difference variables,
where $m$ is the rank of the symbolic support matrix of the super essential system.
%
%the of a Laurent
%transformally essential system, where,  can be simplified to the system consisting of $m+1$ polynomials in $m$ difference variables and the sparse difference resultant of the simplified system is the same to the one of original system.
Moreover, a new order bound of sparse difference resultant is given. Then we propose an efficient algorithm to
compute sparse difference resultant, {\color{black}which is based on the
fact that sparse difference resultant} is shown to be the
sparse resultant for a certain generic algebraic polynomial system.
It starts with the given sparse difference polynomial system and
directly obtain a strong essential polynomial system of the original
system, then one can regard it as sparse algebraic polynomial system
and use the algorithm in \citep{emiris} to construct the matrix
representation whose determinant is the required sparse difference
resultant.
{Furthermore, the computations of finding the strong essential polynomial system are compiled as the function \textbf{SDResultant} with
Mathematica and then the mixed subdivision algorithm is called to search for the sparse resultant.}
%{Furthermore, the whole computations of the algorithm are implemented in computer algebra system where the
%compiled function \textbf{SDResultant} with
%Mathematica generates the strong essential polynomial system and its sparse algebraic resultant is computed by the mixed subdivision algorithm in \citep{emiris}.}

The rest of the paper is arranged as follows. In Section 2, we
review some preliminary results which contains definitions and theorems of sparse resultant and sparse difference resultant. Section 3 concentrates on the main results of the paper involving  the theoretical preparation of the algorithm, algorithm implementation and illustrated examples. The last section concludes the results.

\section{Preliminaries}\label{SECT:preliminaries}
\subsection{{\color{black}On the sparse resultant}}
We first introduce several basic notions and properties {\color{black} of sparse
resultant} which are needed in the algorithm. We refer to
\citep{gelfand,sturmfels,sturmfels2,gao-2015,emiris} for more details.

Let $\mathcal{B}_0,\ldots,\mathcal{B}_n$ be finite subsets of
$\mathbb{Z}^n$. Assume ${\bf 0}\in\mathcal{B}_i$ and
$|\mathcal{B}_i|\geq 2$ for each $i$, where here and below, the symbol $|S|$ denotes the cardinality of the set $S$. For algebraic indeterminates
$\X=\{x_1,\ldots,x_n\}$ and
$\alpha=(\alpha_1,\ldots,\alpha_n)\in\mathbb{Z}^n$, denote
$\X^{\alpha}=\prod_{i=1}^nx_i^{\alpha_i}$. Let
\begin{equation}\label{eq-algsparse}
\BF_i(x_1,\ldots,x_n)=c_{i0}+\sum_{\alpha\in\mathcal{B}_i\backslash
\{{\bf 0}\}}c_{i\alpha}\X^{\alpha}\,(i=0,\ldots,n)
\end{equation}
{\color{black}be a generic sparse Laurent polynomial system and $c_{i\alpha}$ are treated as parameters}. We call $\mathcal{B}_i$ the
support of $\BF_i$ and
$\omega_i=\sum_{\alpha\in\mathcal{B}_i}c_{i\alpha}\alpha$ is called
the {\em symbolic support vector} of $\BF_i$. The  convex hull of $\mathbb{R}^n$ containing $\mathcal{B}_i$ is called the
{\em Newton polytope} of $\BF_i$.
For any subset $I\subset\{0,\ldots,n\}$, the matrix $M_{I}$ whose
row vectors are $\omega_i\, (i\in I)$ is called the {\em symbolic
support matrix} of $\{\BF_i:i\in I\}$.
 Denote $\bc_i=(c_{i\alpha})_{\alpha\in\mathcal{B}_i}$,
$\bc_I=\cup_{i\in I}(\bc_i)$ and by $\rk(M_I)$ the rank of matrix $M_I$.

\begin{define}{\rm(\citep[Definition 61]{gao-2015})}
Follow the notations introduced above.
\begin{itemize}
%\item A collection of  $\{\mathcal{B}_i\}_{i\in\text{
%J}}$, or $\{\BF_i\}_{i\in\text{J}}$ of the form
%(\ref{eq-algsparse}), is said to be {\rm algebraically independent}
%if $\trdeg\,\mathbb{Q}(\bc)(\BF_i-c_{i0}:\,i\in\text{
%J})/\mathbb{Q}(\bc)=|\text{J}|$. Otherwise, they are said to be {\rm
%algebraically dependent}.
\item  A collection of  $\{\BF_i\}_{i\in I}$   is
{\em weak essential} if $\rk(M_I)=|I|-1$.
\item A collection of $\{\BF_i\}_{i\in I}$  is
{\em essential} if $\rk(M_I)=|I|-1$ and for each proper subset
$\text{J}$ of $\text{I}$, $\rk(M_J)=|J|$.
\end{itemize}
\end{define}

{\color{black}Note that we sometimes call system (\ref{eq-algsparse}) algebraically essential instead of essential in order to differentiate it with difference case.}
A polynomial system $\{\BF_i\}_{i\in I}$ is weak essential  if and
only if $(\BF_i:i\in I)\cap \Q[\bc_I]$ is of codimension one {\color{black} in $\Q[\bc_I]$}~\citep{emiris}. In
this case, there exists an irreducible polynomial $\SR\in\Q[\bc_I]$
such that $(\BF_i:i\in I)\cap \Q[\bc_I] = (\SR)$ and $\SR$ is called
the {\em sparse resultant} of $\{\BF_i:i\in I\}$. Furthermore, the
system $\{\BF_i\}_{i\in\text{I}}$ is essential  if and only if
$(\BF_i:i\in I)\cap \Q[\bc_I]= (\SR)$ and $\bc_i$ appears
effectively in $\SR$ for each $i\in I$.

Suppose an arbitrary total ordering of $\{\BF_0,\ldots,\BF_n\}$ is
given, say $\BF_0 < \BF_1 < \cdots < \BF_n$. Now we define a total
ordering $\succ$ among subsets of $\{\BF_0,\ldots,\BF_n\}$. For any two
subsets $\mathcal {D} = \{ D_0,\ldots, D_s \}$ and $\C =
\{C_0,\ldots, C_t \}$ where $D_0 > \cdots > D_s$ and $C_0 >\cdots>
C_t$, $\mathcal{D}$ is said to be of {\em higher ranking }than
$\C$, denoted by $\mathcal{D}\succ\C$, if 1) there exists an $i\leq
\min(s,t)$ such that $D_0 = C_0, \ldots, D_{i-1} = C_{i-1}$, $D_i >
C_i$ or 2) $s>t$ and $D_i = C_i\,(i=0,\ldots,t)$. Note that if
$\mathcal {D}$ is a proper subset of $\C$, then $\C\succ \mathcal
{D}$.
Thus for the system $\BF=\{\BF_i:i=0,\ldots,n\}$ given in
(\ref{eq-algsparse}), if $\rk(M_\BF) \le n$, then $\BF$ has an
essential subset with minimal ranking.

{
\begin{lemma}{\rm(\citep[Lemma 65]{gao-2015})}\label{lm-vess}
Suppose $\BF_I=\{\BF_i:i\in I\}$ is an essential system. Then there
exists an $I'\subset \{ 1,\ldots,n \}$ with $ |I'| = n-|I|+1$, such that by setting $ x_i , i\in I'$ to $1$, the specialized system
$\widetilde{\BF_I} = \{\widetilde{\BF}_i:\,i\in I \}$ satisfies \\
$(1)$ \,\, $\widetilde{\BF_I}$ is still essential. \\
$(2)$ \,\, $\rk(M_{\widetilde{\BF_I}}) = |I|-1 $ is the number of variables in $\widetilde{\BF_I}$. \\
$(3)$ \,\, $(\BF_I)\cap \Q[\bc_I] = (\widetilde{\BF_I})\cap
\Q[\bc_I]$,
where $\widetilde{\BF_I} = \BF_i |_{x_i=1,i\in I'}  $.
\end{lemma}
}

%\begin{definition}
 An essential system $\{\BF_i\}_{i\in I}$  is said to be
\emph{variable-essential} if there are only $|I|-1$ variables
appearing effectively in $\BF_i$. Clearly, if
$\{\BF_i:i=0,\ldots,n\}$ is essential, {\color{black} then by Lemma \ref{lm-vess}} it is
variable-essential.
%In this case, $\{\BF_i\}_{i\in I}$ is an essential system with $|I|$ polynomials and $|I|-1$ variables.
%\end{definition}
We call a variable-essential system $\BF=\{\BF_i:i=0,\ldots,n\}$
{\em strong essential} {\color{black}if there exists }an invertible variable
transformation $x_1 = \prod_{j=1}^n z_{j}^{m_{1j}},\, \ldots,\, x_n =
\prod_{j=1}^n z_j^{m_{nj}}$  such that  the image $\mathbb{G}$ of
$\BF$ under the above transformation is a generic sparse system satisfying:
$(1)$  $\mathbb{G}$ is essential.
$(2)$  $\rm{Span}_\mathbb{Z} (\mathcal {B}) = \mathbb{Z}^n$,
where $\mathcal {B}$ is
the set of all supports of $\mathbb{G}$.
%Note that condition (2) is a basic
%requirement for studying sparse resultant in the literatures and a
%strong essential system defined here is just an essential system as
%defined in~\citep{sturmfels2,an-2011}.
\subsection{Sparse difference resultant}
This section will review the results associated with sparse
difference resultant, for details please refer to reference
\citep{gao-2015}.

An ordinary difference field $\F$ is a field with a unitary
operation $\sigma$ satisfying $\sigma(a+b)=\sigma(a)+\sigma(b)$ and $\sigma(ab)=\sigma(a)\sigma(b)$ for any $a, b\in\F$.
%and $\sigma(a)=0$ if and only if $a=0$.
We call $\sigma$ the {\em
difference (transforming) operator} of $\F$. If $a\in\F$,  $\sigma(a)$ is called
the transform of $a$ and is denoted by $a^{(1)}$. And for
$n\in\mathbb{Z}^{+}$,  $\sigma^n(a)=\sigma^{n-1}(\sigma(a))$ is
called the $n$-th transform of $a$ and denoted by $a^{(n)}$,  with
the usual assumption $a^{(0)}=a$. By $a^{[n]}$ we mean the set $\{a,
a^{(1)}, \ldots, a^{(n)}\}$.
If {\color{black}$\sigma$} is an isomorphism of a difference field, then the field is
called {\em inversive}. Every difference field has an inversive closure
\citep{cohn}. In this paper, all difference fields are assumed to be inversive with characteristic zero.
% $\F$ is assumed to be reflexive.
%A Typical Example Of Inversive Difference Field Is $\Q(X)$ With
%$\Sigma(F(X))=F(X+1)$ For $F(X)\In\Q(X)$.

{A subset $\mathcal {S}$ of a  difference extension field $\mathcal
{G}$ of $\mathcal {F}$ is said to be {\em transformally dependent}
over $\mathcal {F}$ if the set $\{\sigma^{k}(a)|a\in\mathcal
{S},k\geq0\}$ is algebraically dependent over $\mathcal {F}$, otherwise, is
called {\em transformally independent} over $\mathcal {F}$, or
a family of {\em difference indeterminates} over $\mathcal
{F}$.
We say that $\alpha$ is {\em transformally algebraic }or {\em transformally
transcendental} over $\mathcal {F}$ respectively if $\mathcal{S}$ consists of one element $\alpha$. The maximal
subset $\Omega$ of $\mathcal {G}$ which is transformally
independent over $\mathcal {F}$ is said to be a transformal
transcendence basis of $\mathcal {G}$ over $\mathcal {F}$. We use
$\dtrdeg \,\mathcal {G}/\mathcal {F}$  to denote the {\em difference
transcendence degree} of $\mathcal {G}$ over $\mathcal {F}$, which
is the cardinal number of $\Omega$.

Let $\F$ be an ordinary difference field with a transforming
operator $\sigma$.
Let $\Omega$ be the semigroup of elements generated by $\sigma$.
Let $\Y=\{ y_1,\ldots,y_n \}$ be indeterminants and $\F\{\Y\} = \F[\Omega\Y] $ the difference
polynomial ring, where $\Omega\Y = \{ \sigma^i y_j | i\ge 0, 1\le j\le n \}$ and
$\sigma(\sigma^i y_j) = \sigma^{i+1} y_j$.
%
% and $\F\{\mathbb{Y}\}$ the ring of difference
%polynomials in the difference indeterminates
%$\mathbb{Y}=\{y_1,\ldots,y_n\}$.
{
Let $f$ be a difference polynomial in $\F\{\mathbb{Y}\}$.  The order of $f$
w.r.t. $y_i$ is defined to be the greatest number $k$ such that
$y_{i}^{(k)}$ appears effectively in $f$,  denoted by $\ord(f,
y_{i})$. If $y_{i}$ does not appear in $f$,  then we set $\ord(f,
y_{i})=-\infty$. The {\em order} of $f$ is defined to be $\max_{i}\,
\ord(f, y_{i})$,  that is, $\ord(f)=\max_{i}\, \ord(f, y_{i})$.}
A Laurent difference monomial of
order $s$ is in the form
$\prod_{i=1}^n\prod_{k=0}^s(y_i^{(k)})^{d_{ik}}$ where $d_{ik}$ are
integers which can be negative. A {\em Laurent difference
polynomial} over $\F$ is a finite linear combination of Laurent
difference monomials with coefficients in $\F$.
{A difference ideal $I$ in $\F\{\mathbb{Y}\}$ is an algebraic ideal which is closed under $\sigma$,
i.e., $\sigma(I)\subseteq I$.
If $I$ also has the property that $\sigma(a)\in I$ implies $a\in I$, it is called a {\em reflexive difference ideal}.
The concept of difference ideal and reflexive difference ideal can be generalized to Laurent difference case naturally.}

For every  Laurent difference polynomial $F\in\F\{\Y,\Y^{-1}\}$,
there exists a unique Laurent difference monomial $M$ such that
$M\cdot F$  is the {\em norm form} of $F$, denoted by N$(F)$, which satisfies 1) $M\cdot
F\in\F\{\Y\}$ and 2) for any Laurent difference monomial $T$ with
$T\cdot F\in\F\{\Y\}$, $T\cdot F$ is divisible by $M\cdot F$ as
polynomials. The order and degree of $\norm(F)$ is defined to be the
{\em  order} and {\em degree} of $F$, denoted by $\ord(F)$ and
$\deg(F)$.}

Suppose $\mathcal
{A}_i=\{M_{i0},M_{i1},\ldots,M_{il_i}\}\,(i=0,1,\ldots,n)$ are
finite sets of Laurent difference monomials in $\mathbb{Y}.$
Consider $n+1$ {\em generic Laurent difference polynomials} defined
over $\mathcal {A}_0,\ldots,\mathcal{A}_n$:
\begin{equation} \label{eq-sparseLaurent}
\mathbb{P}_i= u_{i0}M_{i0} + \sum\limits_{k=1}^{l_i}u_{ik}
M_{ik}\,~~~(i=0,\ldots,n),
\end{equation} where all the $u_{ik}$ are  transformally independent over
$\F$.
%The set of exponent vectors $\SP_i =\{ \alpha_{ik}:\,
%k=0,\ldots,l_i\}$ is called the {\em support} of $\P_i$. The number
%$|\SP_i| = l_i +1$ is called the {\em size} of $\P_i$.  Note that
%an exponent vector of $\P_i$
%contains $n(s_i+1)$ elements.
 Denote
\begin{equation} \label{eq-uu}
 \mathbf{u}_i=(u_{i0},u_{i1},\ldots,u_{il_i})\,\,(i=0,\ldots,n) \hbox{ and }
   \mathbf{u}=\bigcup_{i=0}^n\mathbf{u}_i.
\end{equation}
The number $ l_i +1$ is called the {\em size} of $\mathbb{P}_i$ and
$\mathcal{A}_i$ is called the {\em support} of $\mathbb{P}_i$.% ~{\color{black}, $l=\sum_{i=0}^n (l_i+1) $ is called the {\em size} of  $\P=\{\P_0,\ldots,\P_n\}$ }.
To avoid this triviality in the sequel we assume that $l_i\geq1\,(i=0,\ldots,n)$.

\begin{define}{\rm(\citep[Definition 11]{gao-2015})}\label{def-tdes} A set of Laurent difference
polynomials of the form (\ref{eq-sparseLaurent}) is called  {\em
Laurent transformally essential} if {\color{black}for every $\mathbb{P}_i$} there exist
$k_i\,(i=0,\ldots,n)$ with $1\leq k_i\leq l_i$ such that
$\dtrdeg\,\Q\langle\frac{M_{0k_0}}{M_{00}},$
$\frac{M_{1k_1}}{M_{10}},\ldots,\frac{M_{nk_n}}{M_{n0}}\rangle/\Q=n.$
In this case, we also say that $\mathcal{A}_0,\ldots,\mathcal{A}_n$
form a Laurent transformally essential system.
\end{define}
This definition generally means that for any Laurent transformally essential system $\P$, $\dtrdeg$ $(\F\{\bu\}[\P]/\F\{\bu\}) =n$ which implies $[\P]\cap \F\{\bu\}$ has codimension one.
%where $\bu = \{ u_{ij}| 0\le i le n, 0\\e j \le l_i \}$.
Note that, for $n+1$ generic difference
polynomial system with given order and degree larger than zero, the system is
Laurent transformally essential since one may take $M_{i0} = 1$ and $M_{ik_i} = y_i$ for $1 \le i\le n$.
Hence, for a random difference system with $n+1$ Laurent difference polynomials, it satisfies the
Laurent transformally essential properties with probability one.

Let $\mathbbm{m}$ be the set of all difference monomials in $\Y$ and
$[\norm(\P_0),\ldots,\norm(\P_n)]$ the difference ideal generated by
$\norm(\P_i)$ in $\Q\{\Y,\bu_0,\ldots,\bu_n\}$.
%=\{f\in  \Q\{\Y,\bu_0,\ldots,\bu_n\}|\,\exists M\in\mathbbm{m}\, \text{s.t.}\, Mf\in [\norm(\P_0),\ldots,\norm(\P_n)]\}
Let
\begin{eqnarray}
\mathcal{I}_{\Y,\bu}&=&([\norm(\P_0),\ldots,\norm(\P_n)]:\mathbbm{m}),\label{eq-I}\\
\mathcal{I}_{\bu}&=& \mathcal{I}_{\Y,\bu}\cap\Q\{
\bu_0,\ldots,\bu_n\}.\label{eq-IU}
 \end{eqnarray}

Now suppose $\P=\{\P_0,\ldots,\P_n\}$ is a Laurent transformally
essential system.
%
%Denote the difference ideal
%$[\P_0,\ldots,\P_n]\cap\Q\{\bu_0,\ldots,\bu_n\}$ by $\I_{\bu}$.
Since $\I_{\bu}$ defined in (\ref{eq-IU}) is a reflexive prime
difference ideal of codimension one, there exists a unique
irreducible difference polynomial $\SR(\bu)=\SR(\bu_{0},\ldots,\bu_{n})$
$\in\Q\{\bu_{0},\ldots,\bu_{n}\}$ such that $\SR$ can serve as the
first polynomial in each characteristic set of $\I_{\bu}$ w.r.t. any
ranking endowed on $\bu_{0},\ldots,\bu_{n}$~\citep{gao-2015}. %That is, if $u_{i0}$
%appears in $\SR$, then among all the difference polynomials in
%$\I_{\bu}$, $\SR$ is of minimal order in $u_{i0}$ and of minimal
%degree with the same order.
Thus the definition of sparse difference resultant is given as
follows:
\begin{define}{\rm(\citep[Definition 15]{gao-2015})}\label{def-sparse}
The above
$\SR(\bu_{0},\ldots,\bu_{n})\in\Q\{\bu_{0},\ldots,\bu_{n}\}$ is
defined to be the {\em sparse difference resultant} of the Laurent
transformally essential system $\P_0,\ldots,\P_n$, denoted by
$\Res_{\mathcal{A}_0,\ldots,\mathcal{A}_n}$ or
$\Res_{\,\P_0,\ldots,\P_n}$, where $\mathcal{A}_i$ is the support of $\P_i$ for $ i=0, 1, \ldots,n$. When all the $\mathcal{A}_i$ are equal to
the same $\mathcal{A}$, we simply denote it by $\Res_\mathcal{A}$.
\end{define}

By Definition \ref{def-sparse}, the existence of sparse difference resultant of a Laurent
difference polynomial system is attributed to determine
whether the system is Laurent transformally essential, which is a cumbersome procedure. In
\citep{gao-2015}, the authors build a one-to-one correspondence
between a difference polynomial system and so-called symbolic
support matrix and then use the matrix to discriminate whether a Laurent
difference system is transformally essential.

Specifically, let $B_i = \prod_{j=1}^n\prod_{k = 0}^s
(y_{j}^{(k)})^{d_{ijk}}\,(i=1,\ldots,m)$ be $m$ Laurent difference
monomials. Introduce a new algebraic indeterminate $x$ and let
$$d_{ij} = \sum_{k=0}^{s} d_{ijk}x^{k}\,~~~ (i=1,\ldots,m ,j=1,\ldots,n)$$
be univariate polynomials in $\mathbb{Z}[x]$. If $\ord(B_i,y_j) =
-\infty$, then set $d_{ij}=0$.
The vector  $(d_{i1},d_{i2},$ $\ldots,d_{in})$ is called the {\em
symbolic support vector} of $B_i$.
The matrix $M=(d_{ij})_{m\times n}$
%\[M=\left(\begin{array}{cccccccc}
%d_{11} & \,d_{12} & \,\ldots & \,d_{1n} \\
%d_{21} & \,d_{22} & \,\ldots & \,d_{2n} \\
%&  & \ddots & \\
%d_{m1} & \,d_{m2} & \,\ldots & \,d_{mn}
%\end{array}\right) \]
is called the {\em symbolic support matrix} of $B_1,\ldots, B_m$.

Consider the set of generic Laurent difference polynomials defined
in (\ref{eq-sparseLaurent}),
\begin{equation*}
\P_i= u_{i0}M_{i0} + \sum\limits_{k=1}^{l_i}u_{ik}
M_{ik}\,~~~(i=0,\ldots,n).
\end{equation*}
Let $\beta_{ik}$ be the symbolic support vector of { $M_{ik}/M_{i0}, k=1,\ldots,l_i$}.
Then the vector $w_i = \sum_{k=1}^{l_i} u_{ik}\beta_{ik}$ is called
the {\em symbolic support vector} of $\P_i$ and the matrix $M_\P$
whose rows are $w_0,\ldots,w_n$ is called the {\em symbolic support
matrix} of $\P_0,\ldots,\P_n$. { Therefore, we have

\begin{theorem}\label{th-ex}{\rm(\citep[Theorem 31]{gao-2015})} A sufficient and necessary
condition for $\P_0,\ldots,\P_n$ form a Laurent transformally
essential system is the rank of $M_\P$ is equal to $n$.
\end{theorem}}

Furthermore, we can use the symbolic support matrix to {\color{black}determine} certain
$\P_i$ such that their coefficients will not occur in the sparse
difference resultant, { which leads to the following definition:}
\begin{define}{\rm(\citep[Definition 33]{gao-2015})} \label{sup-ess}
Let $\TT\subset\{0,1,\ldots,n\}$. Then we call $\TT$ or $\P_\TT$
{\em super-essential} if the following conditions hold: (1)
$\card(\TT) - \rank(M_{\P_\TT}) = 1$ and (2) $\card(\JJ) =
\rank(M_{\P_\JJ})$ for each proper subset $\JJ$ of $\TT$.
\end{define}

The existence of super-essential system of a difference polynomial
system is given by the following theorem \citep{gao-2015}.
\begin{theorem}\label{the-sup}{\rm(\citep[Theorem 34]{gao-2015})}
If $\{\P_0,\ldots,\P_n\}$ is a Laurent transformally essential
system, then for any $\TT\subset\{0, 1,\ldots, n \}$, $\card(\TT) -
\rank(M_{\P_\TT}) \leq 1$ and there exists a unique $\TT$ which is
super-essential. { If $\TT$ is super-essential, then the sparse
difference resultant of $\{\P_0,\ldots,\P_n\}$  involves only the coefficients of $\P_i\, (i \in \TT)$.}
\end{theorem}

Therefore, let $\SR$ be the sparse difference resultant of a
Laurent transformally essential system (\ref{eq-sparseLaurent}).
Then  a strong essential system $\SC$ whose sparse resultant is
equal to $\SR$ can be obtained from (\ref{eq-sparseLaurent}).

{We introduce some notations which are needed to bound the order of $\SR$.}
Let $A=(a_{ij})$ be an $n\times n$ matrix where $a_{ij}$ is an
integer or $-\infty$. A {\em diagonal sum} of $A$ is any sum
$a_{1\sigma(1)}+a_{2\sigma(2)}+\cdots+a_{n\sigma(n)}$ with $\sigma$
a permutation of $1,\ldots,n$. If $A$ is an $m\times n$ matrix with
$k=\min\{m,n\}$, then a diagonal sum of $A$ is a diagonal sum of any
$k\times k$ submatrix of $A$. The {\em Jacobi number}  of a matrix
$A$ is the maximal diagonal sum of $A$, denoted by $\Jac(A)$.

Let $s_{ij}=\ord(\norm(\P_i),y_j)\,(i=0,\ldots,n;j=1,\ldots,n)$ and
$s_i=\ord(\norm(\P_i))$. We call the $(n+1)\times n$ matrix
$A=(s_{ij})$ the {\em order matrix} of $\P_0,\ldots,\P_n$. By
$A_{\hat{i}}$, we mean the submatrix of $A$ obtained by deleting the
$(i+1)$-th row from $A$.
{ Note that when we consider the sparse difference resultant, we compute it in the Laurent
difference polynomial ring. Hence, without loss of generality, we may assume that the normal form of $\P_i$ is itself.
Then, we use $\P$ to denote the set
$\{\norm(\P_0),\ldots,\norm(\P_n)\}$ and by $\P_{\hat{i}}$, we mean the set
$\P\backslash\{\norm(\P_i)\}$.} We call $J_i=\Jac(A_{\hat{i}})$ the {\em
Jacobi number} of the system $\P_{\hat{i}}$, also denoted by
$\Jac(\P_{\hat{i}})$.

\begin{theorem}\label{th-ord-jacobi2}{\rm(\citep[Theorem 51, 74; Lemma 16]{gao-2015})}
Let $\P$ be a Laurent transformally essential system and $\SR$ the
sparse difference resultant of $\P$. Then, $\SR$ is of minimal order in each
$u_{i0}$, and
\[\ord(\SR,\bu_i)=\left\{\begin{array}{lll}
-\infty&& \text{if}\quad\,J_i = -\infty,\\
h_i\leq J_i&& \text{if}\quad\,J_i \geq0.\end{array}\right.\]
Moreover, $\SR \in (\P_0^{[h_0]},\P_1^{[h_1]},\dots,\P_n^{[h_n]})$ in $\Q\{\Y,\Y^{-1},\bu_{0},\ldots,\bu_{n}\}$  .
\end{theorem}

\section{Main results}
In this section, we first present some theoretical results, and then give an efficient algorithm to
compute sparse difference resultant. {\color{black}We analyze the complexity of the algorithm and show the efficiency by three examples.}
\subsection{Theoretical preparations}
{\color{black}By Theorem~\ref{the-sup}, there exists a unique $\TT\subset\{0, 1,\ldots, n \}$ such that
$\TT$ is super essential.
Without loss of generality, we assume $\TT = \{ 0, 1, \ldots, m \}$,
where $m\le n$. The symbolic support matrix of $\P_\TT$ is}
\begin{equation}\label{mpt}
  {M_{\P_\TT}}=\left(
     \begin{matrix}
       w_{0,1}  & w_{0,2} & \ldots & w_{0,n} \\
       w_{1,1}  & w_{1,2} & \ldots & w_{1,n} \\
       \ldots   & \ldots  & \ldots & \ldots  \\
       w_{m,1}  & w_{m,2} & \ldots & w_{m,n} \\
     \end{matrix}
   \right)_{(m+1)\times n},
   \end{equation}
and $\rank(M_{\P_\TT}) = m$. Then we choose a submatrix  from ${M_{\P_\TT}}$ whose column rank is $m$.
Without loss of generality, we assume that the first $m$ columns in  ${M_{\P_\TT}}$ is of rank $m$.
Now, we set $y_i, i=m+1,\ldots,n$, to $1$ in $\P_\TT$ to obtain a new difference polynomial system $\widetilde{\P_\TT}$ whose symbolic support matrix is
\begin{equation}\label{mpt1}
  {M_{\widetilde{\P_\TT}}}=\left(
     \begin{matrix}
       w_{0,1} & w_{0,2} & \ldots & w_{0,m} \\
       w_{1,1} & w_{1,2} & \ldots & w_{1,m} \\
       \ldots  & \ldots  & \ldots & \ldots  \\
       w_{m,1} & w_{m,2} & \ldots & w_{m,m} \\
     \end{matrix}
   \right)_{(m+1)\times m}.
   \end{equation}

Since $\rank(M_{\widetilde{\P_\TT}}) = m$, $\widetilde{\P_\TT}$ is Laurent transformally essential.
Let $\widetilde{w}_i=(w_{i,1}, w_{i,2},$ $ \ldots, w_{i,m})$ be the symbolic support vector of $\widetilde{\P_i}$
for $i=0,\ldots,m$ in $\widetilde{\P_\TT}$, and $\widetilde{\SR}$ be the sparse difference resultant of the system $\widetilde{\P_\TT}$.
{\color{black}Let $\SR$ be the sparse difference resultant of $\P$,
we will show that $\SR = \widetilde{\SR}$.} Before this, we need the following lemma which can be shown by linear algebra
and we omit the proof.

\begin{lemma}\label{matrix-ess}
Let $M$ be a matrix of row codimension one  such that any proper subset of rows is linearly
independent. Let $\widetilde{M}$ be a submatrix of $M$ with the same number of rows such that
$\rank(M) = \rank(\widetilde{M})$. Then, any proper subset of rows of  $\widetilde{M}$ is linearly independent.
\end{lemma}

\begin{lemma}\label{wpt-ess}
$\widetilde{\P_\TT}$ forms a super essential system.
\end{lemma}
\begin{proof}
%By Theorem~\ref{th-ex}, $\widetilde{\P_\TT}$ is a transformally essential system since the rank of $M_{\widetilde{\P_\TT}}$ is $m$.
%Assuming that $\widetilde{\P_\TT}$ is not a super essential system, by Theorem~\ref{the-sup}, there exists $\TT' \subsetneq \TT$
%such that $\widetilde{\P_{\TT'}}$ is super essential, which means that the symbolic support vectors of $\widetilde{\P_{\TT'}}$
%are linear dependent.
% Let $\widetilde{w}_i$ be the symbolic support vector of $\widetilde{\P}_i, i\in \TT'$  and thus $\sum_{i\in\TT'} f_i\widetilde{w}_i = 0$, where $0\ne f_i \in \Q\{\bu_{0},\ldots,\bu_{n}\}[x]$ for any $i\in\TT'$. %Also the equality $\sum_{i\in\TT} f_i\widetilde{w}_i = 0$ holds for $0\ne f_i \in \Q\{\bu_{0},\ldots,\bu_{n}\}[x]$.
%
%On the other hand, since $\P_\TT$ is super essential, we have $\sum_{i\in \TT} g_i\,w_i = 0$, where $0\ne g_i\in \Q\{\bu_{0},\ldots,\bu_{n}\}[x]$ for any $i\in\TT$, which implies $\sum_{i\in \TT} g_i\,\widetilde{w}_i = 0$.
%Since $\TT' \ne \TT$, we find that, for $i\in\TT$, $\widetilde{w}_i$ satisfy two linear independent relations over $\Q\{\bu_{0},\ldots,\bu_{n}\}[x]$,
%which means $\rank(M_{\widetilde{\P_\TT}}) \le  (m+1)-2 = m-1$. It contradicts to $\rank(M_{\widetilde{\P_\TT}})  = m$.
Let $M = M_{\P_\TT}$ and $\widetilde{M} = M_{\widetilde{\P_\TT}}$. Apply Lemma~\ref{matrix-ess} to $M$ and $\widetilde{M}$, we have
$\widetilde{\P_\TT}$ forms a super essential system.
\end{proof}

%Let $\widetilde{\mathcal {P}}$ be the algebraically esstential system obtained by algorithm 1(old version). Let
In analogy with the algebraic case, we define a total ordering of
$\widetilde{\P_\TT}^{\infty} = \{ \sigma^j\widetilde{\P_i}, 0\le i\le m, j\ge 0 \}$,
$\sigma^j\widetilde{\P_i} < \sigma^k\widetilde{\P_l}$ if and only if either $i<l$ or $i=l$ and $j<k$.
Then, this ordering can be extended to a total ordering among the subsets of
$\widetilde{\P_\TT}^{\infty}$ as defined in Section~\ref{SECT:preliminaries}.
{ Consider all possible systems of the form
\begin{equation}\label{eq-algess}
\begin{array}{lllll}
\widetilde{\mathcal {P}} = \big\{ & \sigma^{i_{01}}\widetilde{\P_0}, & \sigma^{i_{02}}\widetilde{\P_0}, &\ldots, & \sigma^{i_{0s_0}}\widetilde{\P_0}, \\
  & \sigma^{i_{11}}\widetilde{\P_1}, & \sigma^{i_{12}}\widetilde{\P_1}, &\ldots, & \sigma^{i_{1s_1}}\widetilde{\P_1}, \\
  & \ldots,  & \ldots,  & \ldots,  & \ldots, \\
  & \sigma^{i_{m1}}\widetilde{\P_m}, & \sigma^{i_{m2}}\widetilde{\P_m}, &\ldots, &\sigma^{i_{ms_m}}\widetilde{\P_m}  \big\},
 \end{array}
 \end{equation}
which are algebraically essential and choose the one that is minimal with respect to the above introduced ordering.
Suppose that $\widetilde{\mathcal {P}}$ is an algebraically essential system with minimal ordering, which always exists due to the proof
of Theorem~68 in \citep{gao-2015}.

{ Note that, in order to define the symbolic support vector of $P\in \widetilde{\mathcal {P}}$, we may assume the symbolic support vector of the variable $\sigma^i y_j$ is $e_{i,j}$ where $e_{i,j}$ are linearly independent. Then, the symbolic support vector of any monomials in $P\in \widetilde{\mathcal {P}}$ can be defined. Since $\sigma (\sigma^i y_j) = \sigma^{i+1} y_j$, we may set $e_{i,j} = x^i e_{j}$ for any $i$ and $1\le j\le m$, where $e_1, \ldots, e_m$ form a standard basis in $\Q^m$, then  $x^i e_j$ are linear independent over $\Q$. Since $\sigma^k u_{i,j}$ is transcendental over $\{\sigma^p u_{i,j}, \sigma^l y_s | p<k, 0\le i\le n, 1\le j\le l_i , l\ge 0, 1\le s\le n \}$, we have that the symbolic support vector of $P=\sigma^i \widetilde{\P_j} $ can be defined as $x^i \widetilde\omega_j$, where $\widetilde\omega_j$ is the symbolic support vector of $\widetilde{\P_j}$. Though we treat the polynomials in $\widetilde{\mathcal {P}}$ as difference polynomial or algebraic polynomial, the linear dependence of the symbolic support vector of the polynomials in $\widetilde{\mathcal {P}}$ will not change in some sense. }

Let
\begin{equation}
\begin{array}{lllll}
 \mathcal{P} = \big\{  & \sigma^{i_{01}}{\P_0}, & \sigma^{i_{02}}{\P_0}, &\ldots, &\sigma^{i_{0s_0}}{\P_0}, \\
  & \sigma^{i_{11}}{\P_1}, & \sigma^{i_{12}}{\P_1}, &\ldots, &\sigma^{i_{1s_1}}{\P_1}, \\
  & \ldots,  & \ldots,  & \ldots,  & \ldots, \\
  & \sigma^{i_{m1}}{\P_m}, & \sigma^{i_{m2}}{\P_m}, &\ldots, &\sigma^{i_{ms_m}}{\P_m}  \big\},
   \end{array}
 \end{equation}
and we have the following result. }

%By Lemma \ref{lm-vess}, if $\mathcal{P}$ is an algebraically essential system, so does $\widetilde{\mathcal {P}}$. Below we show that the converse works.
\begin{lemma}\label{lm-apess}
If $\widetilde{\mathcal {P}}$ is algebraically essential, then $\mathcal{P}$ forms an algebraically essential system.
\end{lemma}
\begin{proof}
{Since $\widetilde{\mathcal{P}}$ is an algebraically essential system, we have that the row corank of
its symbolic support matrix $M_{\widetilde{\mathcal{P}}}$ is $1$.
Since $\widetilde{\mathcal {P}}$ is algebraically essential,
we have $\sum_{\sigma^i\widetilde{\P_j}\in \widetilde{\mathcal{P}}} f_{ij} \widetilde{\omega_{ij}} = 0$, where
$0\ne f_{ij}\in \Q\{\bu_0,\ldots,\bu_n\}$ for any $i,j$ and
$\widetilde{\omega_{ij}}$ is the symbolic support vector of $\sigma^i\widetilde{\P_j}$.
Since $\widetilde{\omega_{ij}} = x^i\widetilde{\omega_{j}}$,
where $\widetilde{\omega_{j}}$ is the symbolic support vector of $\widetilde{\P_j}$,
 we have
$\sum_{\sigma^i\widetilde{\P_j}\in \widetilde{\mathcal{P}}} f_{ij} x^i\widetilde{\omega_{j}} = 0$.
%Since $M_{\mathcal{P}}$ has full rank,

Assume that  $\mathcal{P}$
does not form an algebraically essential system. Then its symbolic support matrix $M_{\mathcal{P}}$ has full rank.
Thus $\sum_{\sigma^i{\P_j}\in {\mathcal{P}}} f_{ij} \omega_{ij} \ne 0$, where $\omega_{ij}$ is the symbolic support vector of $\sigma^i{\P_j}$, or  equivalently,
$\sum_{\sigma^i{\P_j}\in {\mathcal{P}}} f_{ij} x^i\omega_{j} \ne 0$ where $\omega_{j}$ is the symbolic support vector of $\P_j$.
Now we consider the rank of $M_{\P_{\TT}}$. By Lemma~\ref{wpt-ess}, $\widetilde{\P_{\TT}}$ is super essential,
each $m\times m$ submatrix of $M_{\widetilde{\P_{\TT}}}$ is of full rank.
Since $\sum_{\sigma^i\widetilde{\P_j}\in \widetilde{\mathcal{P}}} f_{ij} x^i\widetilde{\omega_{j}} = 0$ and
 $\sum_{\sigma^i{\P_j}\in {\mathcal{P}}} f_{ij} x^i\omega_{j} \ne 0$,
 through a row transformation for $M_{\P_{\TT}}$
 over $\Q\{\bu_{0},\ldots,\bu_{n}\}[x]$, we can obtain a row vector of form
 $(0,\ldots, 0, h_{m+1},\ldots, h_n)$ and $h_{m+1},\ldots, h_n$ are not all zeros.
 Hence, the rank of $M_{\P_{\TT}}$ is $m+1$.
 This contradicts to the fact that $\P_{\TT}$ is a Laurent transformally essential system. Hence, $\mathcal{P}$ forms an algebraically essential system. }
\end{proof}

\begin{theorem} \label{sr=wsr}
With above notations, $\SR = \widetilde{\SR}$ up to a multiplicative constant.
\end{theorem}
\begin{proof}
{By the definition of $\widetilde{\mathcal{P}}$, we have that $\widetilde{\SR}$ is the sparse
resultant of $\widetilde{\mathcal{P}}$ since it has the lowest order in $\bu_n$ due to the ordering introduced before.
By Lemma~\ref{lm-apess}, $\mathcal{P}$ is an algebraically essential system. Let $\SR'$ be the
sparse resultant of $\mathcal{P}$, then $\SR'\in( \widetilde{\mathcal{P}} )$ since $\widetilde{\mathcal{P}}$
is obtained by setting $y_{j}$ to $1$ in $\mathcal{P}$ for $j=m+1,\ldots, n$.
Hence
$ \SR' | \widetilde{\SR}$. Then, $\widetilde{\SR} = \SR'$ up to a multiplicative constant since they are irreducible.

Now, we show that $\SR = \SR'$ up to a multiplicative constant.
{ By the definition of $\widetilde{\mathcal{P}}$ and
Theorem~\ref{th-ord-jacobi2},
$\SR'$ has the lowest order in $\big[\widetilde{{\P_{\TT}}}\big]_L\cap \Q\{\bu_0,\ldots,\bu_m\}$ for each $u_{i0}$, where $[ \P_{\TT}]_L$ is the ideal generated by $\P_{\TT}$ in Laurent difference polynomial ring   $\Q\{\Y,\Y^{-1},\bu_{0},\ldots,\bu_{m}\}$ .
We claim that
$\SR'$ has the lowest rank in $[{{\P_{\TT}}}]_L\cap \Q\{\bu_0,\ldots,\bu_m\}$ for each $u_{i0}$.
If it is not the case, then $\SR$ has a lower rank than $\SR'$ in $u_{i0}$ for some $i$.
Since $[\P_{\TT}]_L\cap \Q\{\bu_{0},\ldots,\bu_{m}\} \subset [\widetilde{\P_{\TT}}]_L\cap \Q\{\bu_{0},\ldots,\bu_{m}\}$, then
$\SR\in [\widetilde{\P_{\TT}}]_L\cap \Q\{\bu_{0},\ldots,\bu_{m}\}$.
It contradicts to the fact that $\SR'$ has the lowest rank in
$[\widetilde{{\P_{\TT}}}]_L\cap \Q\{\bu_0,\ldots,\bu_m\}$ for each $u_{i0}$.} Then by the uniqueness of sparse difference resultant,
$\SR = \SR'=\widetilde{\SR} $ up to a multiplicative constant. }
\end{proof}

{\color{black}Theorem \ref{sr=wsr} provides a simple way to compute sparse difference resultants of a Laurent transformally essential system where the unique super essential system is further simplified.}
Now we can give some new order bounds of the sparse difference resultant $\SR$ which is obviously true according to Theorem~\ref{th-ord-jacobi2} and \ref{sr=wsr}.
Let $A_{\P_{\TT}}$ be the order matrix of the system $\P_{\TT}$ and $A_{\widetilde{\P_{\TT}}}$ the order matrix of
$\widetilde{\P_{\TT}}$ such that the $(m+1)\times m$ matrix $M_{\widetilde{\P_{\TT}}}$ is of rank $m$.
Then, an order bound of $\SR$ is just the order bound of $\SR'$, and by Theorem~\ref{th-ord-jacobi2}, we have

\begin{prop}
The order bound of $\SR$ for each set $\bu_i, 0\le i\le m$  can be bounded by $\Jac( A_{\widetilde{\P_{\TT}}})$,
where $\widetilde{\P_{\TT}}$ forms a Laurent transformally essential system, or equivalently,
the $(m+1)\times m$ matrix $M_{\widetilde{\P_{\TT}}}$ is of rank $m$.
\end{prop}

By above proposition, one can get the order bound by the minimal Jacobi number of the corresponding $m\times m$ full rank sub-matrices of $A_{\P_{\TT}}$.
Furthermore, we have
\begin{prop}\label{prop-jacobi}
Let $M_{\widetilde{\P_{\TT}}}$ be an $(m+1)\times m$ full rank sub-matrix of $M_{\P_{\TT}}$. Let $f_i$ be the greatest common divisor
of the $i$-th column of $M_{\widetilde{\P_{\TT}}}$ for $i=1,\ldots m$. Then the order bound of $\SR$ for each set $\bu_i$
can be bounded by $J_i = \Jac\big( (A_{\widetilde{\P_{\TT}}})_{\hat{i}}\big) - \sum_{i=1}^m{\deg(f_i)}$.
\end{prop}
\begin{proof}
Let $\widetilde{ \P_{\TT} }$ be a super essential system and  $M_{\widetilde{\P_{\TT}}}$ its symbolic support matrix.
We denote by $\c_i$ the $i$-th column of $M_{\widetilde{\P_{\TT}}}$ and $\c_i' = \c_i/f_i$.
Then $\c_i$ and $\c_i'$ are linearly dependent.
 Let  $M_{\widehat{\P_{\TT}}}$ be the
$( m+1) \times ( m+1)$ matrix by adding a last column $\c_1'$ to
$M_{\widetilde{\P_{\TT}}}$ and $\widehat{\P_{\TT}}$ its corresponding difference system
{\color{black}by introducing a new difference indeterminant}.
Then by Theorem~\ref{sr=wsr},
$\widehat{\P_{\TT}}$ and $\widetilde{ \P_{\TT} }$ have the same sparse difference resultant. Let
 $M_{\widetilde{ \P_{\TT} }'}$ be the sub-matrix of $M_{{\widehat{\P_{\TT}}}}$ by deleting the first column and
 $\widetilde{ \P_{\TT} }'$ its corresponding difference system. By Theorem~\ref{sr=wsr} again,
$\widetilde{ \P_{\TT} }'$  and $\widetilde{ \P_{\TT} }$ have the same sparse difference resultant.
Inductively, one may take an $(m+1)\times m$ matrix $M_{\overline{\P_\TT}}$ with $\c_i'$ as its $i$-th column and its corresponding
difference system $\overline{\P_\TT}$, such that $\overline{\P_\TT}$ and $\P_\TT$ have the same sparse difference resultant. Thus the order bound of $\SR$ for each set $\bu_i$  equals the $i$-th Jacobi number of the order matrix $A_{\overline{\P_\TT}}$, which is $\widetilde{J_i}- \sum_{i=1}^m{\deg(f_i)}$
where $\widetilde{J_i} = \Jac( (A_{\widetilde{\P_{\TT}}})_{\hat{i}})$.
\end{proof}

In what follows, we state a proposition which will accelerate the algorithm to search for a simple algebraic polynomial system induced by $\widetilde{\P_\TT}$.
\begin{prop}\label{prop-1}
Let $\mathbb{\widehat{P}}=\{\P_0^{[J_0]},\P_1^{[J_1]},\ldots,\P_m^{[J_m]}\}$
be a polynomial system obtained from a transformally essential difference system,
$\mathcal{D}_{\mathbb{\widehat{P}}}$ the algebraic symbolic support matrix of $\mathbb{\widehat{P}}$,
Jacobi numbers $J_i$ %of the order matrix $A_{\mathbb{P}_{\TT}}$
be given as above. Then its algebraic essential
system is contained in
$\{\P_0^{[J_0-p]},\P_1^{[J_1-p]},$ $\dots, \P_m^{[J_m-p]}\}$, where
$p=|\mathbb{\widehat{P}}|-\mbox{rank}(\mathcal{D}_{\mathbb{\widehat{P}}})-1$.
\end{prop}
\begin{proof}
Let $\mathbb{\overline{P}} = \{\P_0^{[J_0-p]},\P_1^{[J_1-p]},\dots,\P_m^{[J_m-p]}\}$ and
$\mathcal{D}_{\mathbb{\overline{P}}}$ the algebraic symbolic support matrix
of $\mathbb{\overline{P}}$. We only need to show that $\mathcal{D}_{\mathbb{\overline{P}}}$ is not of full row rank.
If it is not the case, let $\SR$ be the sparse difference resultant of the original system, then there exist
$Q_P, T_{ij}\in \F\{\Y,\Y^{-1}\}$ such that
$\SR = \sum_{P\in \mathbb{\overline{P}}} Q_P P + \sum_{i=0}^m \sum_{j=1}^{k_i}T_{ij}\sigma^{J_i-p+j} \P_i$.
We claim that there exists an $i$, such that $\ord(\SR,\bu_i) > J_i - p $. If this is not the case,
 substituting $ \sigma^{J_i-p+j} u_{i0} $ by $ \sigma^{J_i-p+j} (-\sum_{k=1}^{l_i}u_{ik}
M_{ik}/M_{i0}) $ in both sides of the above equation,
 we have  $\SR = \sum_{P\in \mathbb{\overline{P}}} \overline{Q}_P P$ for some $\overline{Q}_P \in \F\{\Y,\Y^{-1}\}$,
 which contradicts to the fact that $\mathcal{D}_{\mathbb{\overline{P}}}$ is of full row rank.
Hence the claim is true.
{Then, the algebraic symbolic support matrix of
$\mathbb{P'} = \{\P_0^{[J_0]},\P_1^{[J_1]},\ldots, \P_{i-1}^{[J_{i-1}]}, \P_i^{[J_i-p]}, \P_{i+1}^{[J_{i+1}]},\ldots, \P_m^{[J_m]}\}$
is of full rank by Theorem~\ref{th-ord-jacobi2}.}
Comparing $\mathbb{P'}$ with $\mathbb{\widehat{P}}$, we find that the co-rank of $\mathcal{D}_{\mathbb{\widehat{P}}}$
is no more than $p$ which contradicts to the definition of $p$.
\end{proof}

\subsection{{A new algorithm for sparse difference resultant}}
{ Based on the new bounds and propositions for the sparse difference resultant,
we propose an improved algorithm to compute sparse
difference resultant for any given Laurent transformally essential difference polynomial system. The
algorithm is motivated by the fact that sparse difference resultant
is equal to the sparse resultant of a strong essential
polynomial system which is derived from the original difference
polynomial system~\citep{gao-2015}. Thus one can transform the obtainment of sparse
difference resultant to compute sparse resultant which has
mature algorithms such as subdivision method initiated by Canny and
Emiris \citep{emiris}.}

Therefore, the algorithm is divided into two parts. The first part
is to find the strong essential polynomial system. The main strategy
for this part is to use the symbolic support matrix of the given
difference polynomial system to determine the existence of sparse
difference resultant.
{\color{black} If yes, obtain the unique super-essential system,
and simplify the super-essential system based on Theorem \ref{sr=wsr},
then use algebraic tools to find the strong essential system.}
% and if yes, to obtain the unique super-essential system, and then simplify the super-essential system based on
% Theorem \ref{sr=wsr} and use algebraic tools to find the strong essential system.
 The second
one is to use the mixed subdivision method to construct matrix
representation of sparse resultant of the strong essential system
%whose determinant is {\color{red} a multiple of the required sparse difference resultant, even up to a sign.}
{which provides the required sparse difference resultant up to a sign.}

In order to present a better understanding of the whole procedure, we give the flow chart of \textbf{Algorithm 1} in Figure 1.\\

%%%%%%%%%%%%%%%%%%%%%%%%%%%%%%

\begin{minipage}{15cm}

\[\begin{CD}
\framebox[12em]{A difference system $\mathbb{P}$}
@.\framebox[15em]{Sparse difference resultant of $\mathbb{P}$}
\\@V\text{rank}(D_{\mathbb{P}})=n VV  @A\text{Algebraic sparse} A \text{resultant of
$\mathcal {\widehat{P}}$ }A\\
\framebox[12em]{\parbox{10em}{Laurent transformally essential system
$\mathbb{P}_{I}$}}@.\framebox[15em]{Simplification of $\mathcal {P}$ to obtain $\mathcal {\widehat{P}}$}\\
@V \text{rk}(D_{\mathbb{P}_I}) =
|I|-1 V J\subseteq I,|J|= \text{rk}(D_{\mathbb{P}_J})V  @A{\scriptsize\text{Assignments}}A
\text{variable transformation}A \\
\framebox[12em]{Super-essential system $\mathbb{P}_{\TT}$
}@. \framebox[15em]{\parbox{14em}{Find algebraic essential system
with minimal ranking to obtain $\mathcal {P}$}}\\
@V \text{Set variables} V \{y_{m+1},\dots,y_n\} \text{to one} V  @AAA\\
%Set variables  to one V  @AAA \\
%@V Set variables $\{y_{m+1},\dots,y_n\}$ to one V
% @AA Set $n-\mbox{rank}(\mathcal{D}_{\mathcal {P}})$ variables to one A\\
 %
\framebox[12em]{\parbox{10em}{A new super-essential difference system} }
@>>> \framebox[15em]{\parbox{14em}{{Algebraic system $\mathbb{\widehat{P}}$
 obtained by order bounds}}}\\
%%
%@V\mathbb{\widehat{P}}=\{P_0^{[J_0-p]},P_1^{[J_1-p]},V\dots,~ P_m^{[J_m-p]}\}V @AA A\\
%%
%\framebox[13em]{\parbox{10em}{} }
%@>>> \framebox[15em]{\parbox{14em}{Find algebraic essential system
%with minimal ranking  $\mathcal {P}$}}\\
 \end{CD}\]

 \begin{center}
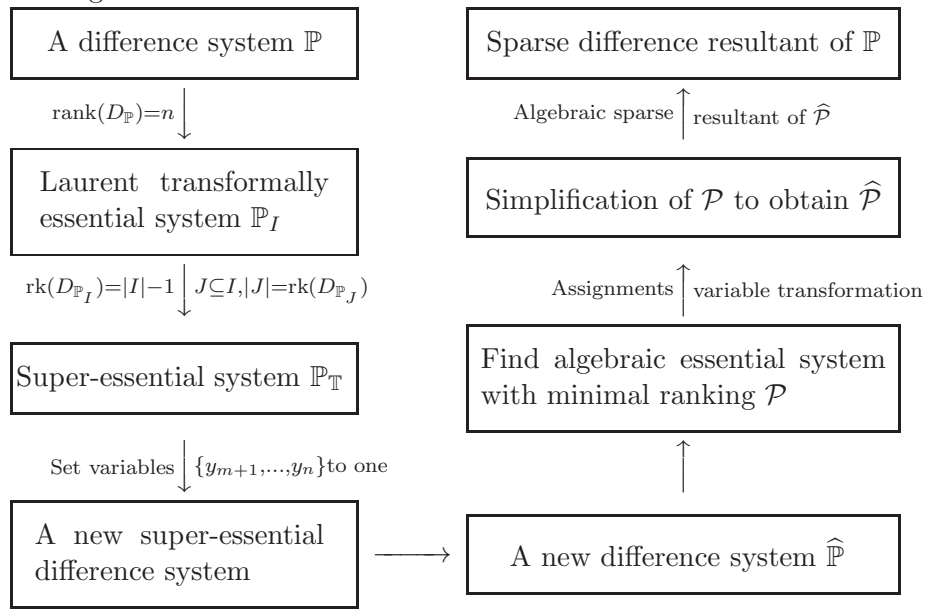
 Figure 1. Flow chart of the algorithm
\end{center}
\end{minipage}

\begin{algorithm}[H]\label{alg-dresl}
  \caption{\bf ----- SDResultant($\P$)} \smallskip
  \Inp{A generic {Laurent} difference system $\P=\{\P_0,\ldots,\P_n\}$.}\\
  \Outp{The sparse difference resultant $\SR$ of $\P$.}\medskip
{
  \noindent

  1. Construct the symbolic support matrix $D_{\P}$ of $\P$.\\
      \SPC If rank($D_{\P})= n$, then proceed to compute SDResultant;\\
      \SPC Else, return ``No SDResultant for $\P$ ''. \\
  2. Set $\TT = \{0,1,\dots,n\}, S = \emptyset$.\\
  3. Let $S = \emptyset$, choose an element $i\in\TT$.\\
   \SPC 3.1. Let $S = S\cup\{ i \}, \TT' = \TT\setminus S, $  \\
   \SPC 3.2  If rank($D_{\P_{\TT'}})= |\TT'|-1$, set $\TT=\TT'$, return to step 3. \\
   \SPC\SPC 3.2.1     Else if $\TT= S$, go to the next step, \\
   \SPC\SPC 3.2.2    else choose an $i\in \TT\setminus S$ go back to step 3.1.\\
     \SPC Note $\P_{\TT}$ is a super-essential system.\\
  4. Assume $\TT=\{0,1,\dots,m\}$.
 {\color{black}
 Compute the rank of  the symbolic support matrix $D_{\P_\TT}$ of $\P_\TT$ by Gauss elimination.
%Denote by $D_{\P_\TT^m}$ the sub-matrix of  $D_{\P_\TT}$ which
 We obtain $(i_1,\ldots,i_m)$ such that
  the {\color{black}$(m+1)\times m$} matrix which corresponds to the
 $i_1$-th, $\ldots, i_m$-th columns of $D_{\P_\TT}$ has rank $m$.
  %and the first $m$ columns of the symbolic support matrix of $\P_\TT$
  % is of rank $m$.
 % Without loss of generality,
  We assume these columns are the first $m$ columns.}
%
%{Denote by $\P_\TT$ the symbolic support matrix and its sub-matrix by $\P_\TT^m$ whose columns correspond to
%the variables $\Y^m=\{y_{i_1},\dots,y_{i_m}\}\subset\Y$. If $\rk(\P_\TT^m)=m$, then set the variables in $\Y\setminus\Y^m$ to 1,}
  {\color{black}
  %Denote by $D_{\P_\TT^m}$ the sub-matrix of $D_{\P_\TT}$ which is obtained by choosing the
  %first $m+1$ columns of $D_{\P_\TT}$.
  Set the variables $\{y_{m+1},\dots,y_n\}$ in $\P_{\TT}$ to 1,
 % Compute an $(m+1)\times m$ sub-matrix $M'$ of $\P_\TT$, such that $\rk(M') = m$.
%  Without loss of generality, assume $M'$ is the first $m$ columns of $\P_\TT$.
   %Set variables $\{y_{m+1},\dots,y_n\}$ in $\P_{\TT}$ to 1,
   we denoted by $\widetilde{\P_{\TT}}$ the new system under this substitution. }\\
  \SPC Compute the order matrix $A_{\widetilde{\P_{\TT}}}$ of $\widetilde{\P_{\TT}}$ and $f_i$ the common
  factor of the $i$-th column of $M_{\widetilde{\P_{\TT}}}$, compute the Jacobi number $J_i=\Jac\big( (A_{\widetilde{\P_{\TT}}})_{\hat{i}}\big) - \sum_{i=1}^m{\deg(f_i)}.$\\
  \SPC Construct a new algebraic system $\mathbb{\widehat{P}}=\{\widetilde{\P}_0^{[J_0]},\widetilde{\P}_1^{[J_1]},\dots,\widetilde{\P}_m^{[J_m]}\}$.\\
  5. Compute the algebraic symbolic support matrix $\mathcal{D}_{\mathbb{\widehat{P}}}$ of $\mathbb{\widehat{P}}$, let $p=|\mathbb{\widehat{P}}|-\mbox{rank}({D}_{\mathbb{\widehat{P}}})-1$.\\
 \SPC  Find the algebraic essential system with minimal ranking  $\mathcal {P}$ from $\widehat{P}_e=\{P_0^{[J_0-p]},P_1^{[J_1-p]},\dots,$ \SPC $P_m^{[J_m-p]}\}$. \\
 6. Select $n-\mbox{rank}({D}_{\mathcal {P}})$ variables in $\mathcal{P}$ to 1, denoted by $\mathcal{\widehat{P}}$.\\
\SPC Take a variable transformation for $\mathcal {\widehat{P}}$ to make it be a strong essential system. \\
  7. Use mixed subdivision algorithm to obtain sparse algebraic resultant of $\mathcal {\widehat{P}}$.\medskip
}

  \noindent /*/\;  In step 4, by {\color{black}Gaussian elimination,} we may obtain the row echelon form. The indices are corresponding to the columns
  of the pivots in the row echelon form.\\
  \noindent /*/\;  Steps $5\sim7$ are performed purely {\color{black}algebraically}.\smallskip
%
%  \noindent /*/\; $\coeff(P,V)$ returns the set of coefficients of $P$ as an ordinary  polynomial in
%  variables $V$.
%\smallskip
\end{algorithm}

{
\begin{theorem}\label{th-alg}
The algorithm is correct.
\end{theorem}
\begin{proof}
The termination of the algorithm is obvious. The correctness of the algorithm is guaranteed by Lemma~\ref{lm-vess},
Theorem~\ref{the-sup}, Theorem~\ref{th-ord-jacobi2}, Lemma~\ref{wpt-ess}, Lemma~\ref{lm-apess}, Theorem~\ref{sr=wsr} and Proposition~\ref{prop-1}.
\end{proof}
}

\subsection{Complexity of the algorithm}
We divide the complexity analysis of \textbf{Algorithm 1} into two parts: the first six steps and the mixed
subdivision algorithm in step 7, and then combine these two parts to estimate the overall complexity.
%The viewpoint is that of worst-case complexity.
In complexity bounds, we sometimes ignore polylogarithmic factors in the parameters appearing in
polynomial factors; this is denoted by $O^*(\cdot )$.

We first recall the complexity analysis of the mixed subdivision algorithm. The key point of mixed subdivision algorithm is to construct Newton matrix whose determinant is a nontrivial multiple of sparse resultant. We refer to \citep{canny-2000} and references therein for the related background materials and constructed techniques of Newton matrix.
The complexity of Newton matrix construction is analyzed in \citep{canny-2000} and recalled as follows.
\begin{theorem}\label{th-111}\citep[Theorem 11.6]{canny-2000}
Given polynomials $g_1,\dots, g_n$,  the algorithm computes an implicit representation of Newton matrix $M$ with worst-case bit complexity
$$O^*\left(|\mathscr{E}|n^{9.5}\mu^{6.5}\log^2d\,\log\frac{1}{\epsilon_l\epsilon_\sigma}\right)$$
where $|\mathscr{E}|$ is the cardinality of the set that indexes the rows and columns of Newton matrix $M$,  $\mu$ is the maximum point cardinality of the $n+1$ supports, $d$ is the maximum degree of any polynomial in any variable, and
$\epsilon_l,\epsilon_\sigma\in(0, 1)$  are  the  error probabilities for the lifting scheme and the perturbation, respectively.
\end{theorem}

Before we estimate the complexity of \textbf{Algorithm 1},
the following lemma is needed for the complexity analysis.
Before we estimate the complexity of \textbf{Algorithm 1},
the following lemma is needed for the complexity analysis.
{
\begin{lemma}\label{lm-rank}
For a symbolic matrix $M = (m_{i,j})$ with $p$ rows and $q$ columns, where $m_{i,j}\in \mathbb{C}[y_1,\ldots,y_k]$. The {\color{black} arithmetic} complexity
to compute the rank of $M$ with probability $1-\epsilon$ is bounded by $\max(p,q)^3$, where $\epsilon$ is the  error probabilities for the rank of the symbolic support matrices.
\end{lemma}}
\begin{proof}
{
Assume the rank of $M$ is $r$. Then, there exists an $r\times r$ sub-matrix $M_r$ of $M$,
such that $\rm{det}(M_r)$ $=g(y_1,\ldots, y_k) \ne 0$. Hence, we set $y_1,\ldots, y_k$ to concrete values
$a_1, \ldots, a_k$ in a given set $S$. {\color{black} By~Schwartz-Zippel Lemma\citep{zippel},  the probability of $g(a_1,\ldots,a_k) \ne 0$ is bounded by $1-d/|S|$, where $d$ is the degree of $g$.}
Now, let $\hat{M}$ be the matrix obtained by substituting $y_i$ by $a_i$ in $M$.
Then, $\rk(\hat{M}) \le \rk(M)$.
If  $g(a_1,\ldots,a_k) \ne 0$, the rank of $\hat{M}$ is $r$ and the time complexity
to compute the rank of $\hat{M}$ is bounded by $O(\max(p,q)^3)$. Denote by $\epsilon$  the  error probabilities for the rank computation for the symbolic support matrices which is depend on $S$, then the {\color{black} arithmetic} complexity
to compute the rank of $M$ with probability $1-\epsilon$ is bounded by $\max(p,q)^3$.}
\end{proof}
{
\begin{remark}
{\color{black} Here, we assume that the constant $a_i$ are taken randomly in a given set $S$, hence the error probabilities for the rank computation depends on the size of the set $S$, the size of the input matrix and the degree of the entries of the matrix. }%
%In this sence, by Schwartz-Zippel Lemma\citep{zippel} for arbitrary choice of $a_i$ the probability of $g(a_1,\ldots,a_k)=0$ is
%bounded by $(d/|S|)^k$ where $d$ is the degree of $g$.
To avoid the probability, one can derive deterministic bounds by using, for instance, test-sets of points\citep{recio-2018}.
\end{remark}
}

Now we give  the complexity of \textbf{Algorithm 1}.

\begin{theorem}
The total complexity of \textbf{Algorithm 1} is bounded by { \color{black}
$$O^*\left(|\mathscr{E}|n^{25.5} s^{16}\mu^{6.5}(\log^2d)\, \epsilon \right), $$ %\epsilon_r\log\frac{1}{\epsilon_l\epsilon_\sigma} \right),$$
} where $|\mathscr{E}|$ is the cardinality of the set that indexes the rows and columns of Newton matrix
 $M$ corresponds to $\mathcal {\widehat{P}}$,
 $\mu$ is the maximum point cardinality of the $n+1$ supports w.r.t. $\mathbb{P}$,
 $d$ is the degree bound  of any polynomial in any variable, and
{\color{black} $\epsilon$ is the error probability of the Algorithm.}
%$\epsilon_l,\epsilon_\sigma,\epsilon_r\in(0, 1)$  are  the  error probabilities for the lifting scheme,
%the perturbation and the rank computations for the symbolic support matrices, respectively.
\end{theorem}

\begin{proof}
We analyze the computational complexity for each step of \textbf{Algorithm 1}.

Step 1.
The first step is to compute the rank of symbolic support matrix $D_{\P}$ whose size is $(n+1)\times n$.
%{\color{black} Note that $|D_{\P}| \le d$ and the number of variables occur in $D_{\P}$ is bounded by $l+1$ where $l=\sum_{i=0}^n l_i$.}
By Lemma~\ref{lm-rank}, {\color{black} the arithmetic complexity is bounded by $O^*(n^3)$ with error probability $\epsilon_1$}.

Step 2.
The complexity of this step can be ignored.

Step 3.
 This step is to find the super essential system of $\P$.
 Consider the worst case which means the super essential system only contains two polynomials.
 Thus we need $(n-1)$ loops from $(n+1)$ to 2. In the $k$-th loop, one need to compute the rank of a
 matrix with size $(n+1-k)\times n$. %and rank of the $(n-k)\times n$ matrix $(n+1-k)$ times,
 Then, by Lemma~\ref{lm-rank}, the {\color{black} arithmetic} computational cost of this step is
$\sum_{k=0}^{n-1}[(n+1-k)^2 n]=O(n^3)$ {\color{black} with error probability $\epsilon_2$}.

Step $4$.
By \citep{moenck-1973}, it needs at most $O( d(\log d)^2)$ steps to compute the gcd of two polynomials in $\mathbb{Z}[x]$,
so the {\color{black} arithmetic} complexity of computing the gcd of $F \in \mathbb{Z}[x]$ is bounded by $(m-1)d(\log d)^2$.
By \citep{jacobi-number}, using Jacobi's algorithm, the {\color{black} arithmetic} complexity of Jacobi number is bounded by $O(n^3)$.

Step 5. The fifth step is to look for the algebraic essential system with minimal ranking  $\mathcal {P}$,
which is similar to the third step except for the big size algebraic symbolic support matrix. And one need to compute the rank of $\mathcal{D}_{\mathbb{\widehat{P}}}$ and the sub-matrix of the symbolic matrix of $\widehat{P}_e$.

We consider the worst case. Denote by $s=\max_i(ord(\P_i))$ the maximal order of $\P_i$.
The worst case is the Jacobi number $J_i = ms$ and $p=0$, then the number of polynomials in the
set $\widehat{P}_e$ is bounded by $(ms+1)(m+1)$. Then the symbolic support matrix of $\widehat{P}_e$ has the
size bounded by $(ms+1)(m+1)\times ((ms+1)(m+1)-1)$,
 and thus the {\color{black} arithmetic} complexity of computing the rank is bounded by $O(m^6s^3)$~{\color{black} with error probability $\epsilon_3$} for each loop by Lemma~\ref{lm-rank}.
In analogy with Step 3, the {\color{black} arithmetic} computational time in this step is bounded by  $O(m^8s^4)$~{\color{black} with error probability $\epsilon_4$}. Hence,
 the {\color{black} arithmetic} complexity in this step is bounded by $O(n^8s^4)$~~{\color{black} with error probability $\epsilon_4$} since $m\le n$.

 Step 6.
 To select $n-\mbox{rank}({D}_{\mathcal {P}})$ variables in $\mathcal{P}$ to 1, we can determine these variables according the result obtained from Step 5. To take a variable transformation for $\mathcal {\widehat{P}}$ to make it be a strong essential system, we only need to compute the Smith normal form of the matrix formed by the support vector of the monomials in $\mathcal {\widehat{P}}$. The monomials are bounded by $(ns+1)\mu$, where $\mu$ is the maximum point cardinality of the $n+1$ supports w.r.t. $\mathbb{P}$, hence the matrix size are bounded by $(ns+1)\mu \times (ns+1)(n+1)$, the {\color{black} arithmetic} complexity of this step is bounded by $O(n^6s^4\mu^2(\log d)^{1+\theta} )$~\citep[Proposition 8.10]{Storjohann2000}, where $0< \theta <1$ is a small positive number.

 Step 7.
 Now, consider the degree and size of $\mathcal {\widehat{P}}$.
 The number of polynomials of $\mathcal {\widehat{P}}$ is bounded by $(ns+1)(n+1)$.
 { Comparing with the original system $\mathbb{P}$, the degree
 keep invariant and the maximum point cardinality of the supports is bounded by $(ns+1)\mu$. Hence, by Theorem \ref{th-111}, the {\color{black} arithmetic} complexity to compute the sparse resultant of $\mathcal {\widehat{P}}$
 is bounded by
 $O^*\left(|\mathscr{E}|n^{25.5} s^{16}\mu^{6.5}\log^2d\,\log\frac{1}{\epsilon_l\epsilon_\sigma}\right)$,}
 where $|\mathscr{E}|$ is the cardinality of the set that indexes the rows and columns of Newton matrix
 $M$ corresponds to $\mathcal {\widehat{P}}$,
 $\mu$ is the maximum point cardinality of the $n+1$ supports w.r.t. $\mathbb{P}$,
 $d$ is the maximum degree of any polynomial $\mathbb{P}_i$ in any variable, and
$\epsilon_l,\epsilon_\sigma\in(0, 1)$  are  the  error probabilities for the lifting scheme and the perturbation, respectively.

{\color{black} Let $\epsilon$ be the total error probability for the whole Algorithm.}
Summarizing the above complexity analysis yields the total complexity. The proof ends.
\end{proof}

\subsection{Implementation and Examples}
{\color{black}\textbf{Algorithm 1} for finding the sparse difference resultant described above
has been implemented in the computer algebra systems {\color{black}Mathematica and Maple.}
The complied function \textbf{SDResultant} with Mathematica outputs the strong essential polynomial system, and then calling mixed subdivision method with Maple in \citep{emiris} gives the required sparse difference resultant  which corresponds to the sparse resultant of the obtained strong essential polynomial system. The interface of  \textbf{SDResultant} only need two arguments: the difference polynomial system and the difference indeterminates.
The function \textbf{SDResultant} will automatically check whether the input difference polynomial system
is Laurent transformally essential or not. If yes, the function \textbf{SDResultant} returns the strong essential polynomial system.}

There exist two obvious merits for \textbf{Algorithm 1}. One merit of \textbf{SDResultant} is that, by characterizing the difference polynomials with the
corresponding symbolic support matrix, it only requires the techniques of linear algebra, such as computing the rank of matrices, row reduction and so on, to discriminate the related conditions and finally output the strong essential polynomial system. Another one is that the algorithm finally gives the matrix representation of sparse difference resultant
which may facilitate to show the properties and explore fast algorithms for sparse difference resultant.
%{ But an obvious drawback is that the implementations are performed on two different computer algebra systems Mathematica and Maple, thus it will be convenient to unify the two parts in the future.}

%Further applications of the SDresultant routine can be found in the
%next section.

 \subsubsection{An artificial example}
{ In this section, we illustrate the  \textbf{Algorithm 1} by an artificial difference polynomial system} $\mathbb{P} =
\{\mathbb{P}_0,\mathbb{P}_1,\mathbb{P}_2, \mathbb{P}_3,\mathbb{P}_4\}$, where $y_{ij} = y^{(j)}_i$ and
%\begin{eqnarray}
%&&\no \mathbb{P}_0=u_{00} + u_{01} y^2_{11} y^2_{21} y_3 + u_{02} y^2_1 y_2 y_3, \\
%&&\no \mathbb{P}_1 = u_{10} + u_{11} y^4_{12} y^4_{22} y^2_{31} + u_{12} y^2_{11} y_{21} y_{31}, \\
%&&\no \mathbb{P}_2 = u_{20} + u_{21} y^2_{11} y^2_{21} y_3 + u_{22} y^2_1 y_2 y_3,\\
%&&\mathbb{P}_3 = u_{30} + u_{31} y_{11} y_3.
%\end{eqnarray}
\begin{eqnarray}
&&\no \hspace{-0.5cm}\mathbb{P}_0=u_{00}+u_{01}\, y_{11}^2 y_{21}^2 y_{31}+u_{02}\, y_{1}^2 y_{2} y_{3}y_4 y_{41},\\
&&\no\hspace{-0.5cm}\mathbb{P}_1=u_{10}+u_{11} \,y_{11}^2 y_{21}^2 y_{31}+u_{12}\, y_{11}^2 y_{21} y_{31} y_{41}y_{42}, \\
   &&\no \hspace{-0.5cm}\mathbb{P}_2=u_{20}+u_{21}\,y_{12}^2 y_{22}^2 y_{32}+u_{22} \,y_{11}^2 y_{21}^2 y_{31} +u_{23}\,y_{1}^2 y_{2} y_{3} y_4y_{41},\\
   &&\no \hspace{-0.5cm}\mathbb{P}_3=u_{30}+u_{31} \, y_{11}y_{21}+u_{32}\, y_{11}^2 y_{21} y_{31} y_{42},\\
   &&\no\hspace{-0.5cm}\mathbb{P}_4=u_{40}+u_{41}\, y_{11} y_{32} y_{41}+u_{42}\, y_{11}^2 y_{22} y_{4}.
\end{eqnarray}

\textbf{ 1. Concrete computations}

The first step is to check whether or not the difference system $\mathbb{P}$ is transformally essential.
The symbolic support matrix of $\mathbb{P}$ is
$$D_{\mathbb{P}}=\left(
\begin{array}{cccc}
 2 x u_{01}+2 u_{02} & 2 x u_{01}+u_{02} & x u_{01}+u_{02} & (x+1) u_{02} \\
 2 x u_{11}+2 x u_{12} & 2 x u_{11}+x u_{12} & x u_{11}+x u_{12} & (x^2+x) u_{12} \\
 2 u_{21} x^2+2 u_{22} x+2 u_{23} & 2 u_{21} x^2+2 u_{22} x+u_{23} & u_{21} x^2+u_{22}
   x+u_{23} & (x+1) u_{23} \\
 x u_{31}+2 x u_{32} & x u_{31}+x u_{32} & x u_{32} & x^2 u_{32} \\
 x u_{41}+2 x u_{42} & x^2 u_{42} & x^2 u_{41} & x u_{41}+u_{42} \\
\end{array}
\right).
$$
%\begin{figure*} [htp]
%\begin{center} {\sc
%$D_{\mathbb{P}}=\left(
%\begin{array}{cccc}
% 2 x u_{01}+2 u_{02} & 2 x u_{01}+u_{02} & x u_{01}+u_{02} & (x+1) u_{02} \\
% 2 x u_{11}+2 x u_{12} & 2 x u_{11}+x u_{12} & x u_{11}+x u_{12} & (x^2+x) u_{12} \\
% 2 u_{21} x^2+2 u_{22} x+2 u_{23} & 2 u_{21} x^2+2 u_{22} x+u_{23} & u_{21} x^2+u_{22}
%   x+u_{23} & (x+1) u_{23} \\
% x u_{31}+2 x u_{32} & x u_{31}+x u_{32} & x u_{32} & x^2 u_{32} \\
% x u_{41}+2 x u_{42} & x^2 u_{42} & x^2 u_{41} & x u_{41}+u_{42} \\
%\end{array}
%\right).
%$}
%%\caption{Complexity of creative telescoping methods (under Hyp.~(H')), together with bounds on output}\label{fig:complexity}
%\end{center}
%\vskip-10pt
%\end{figure*}
%
%
%$$D_{\mathbb{P}}=\left(
%\begin{array}{cccc}
% 2 x u_{01}+2 u_{02} & 2 x u_{01}+u_{02} & x u_{01}+u_{02} & (x+1) u_{02} \\
% 2 x u_{11}+2 x u_{12} & 2 x u_{11}+x u_{12} & x u_{11}+x u_{12} & (x^2+x) u_{12} \\
% 2 u_{21} x^2+2 u_{22} x+2 u_{23} & 2 u_{21} x^2+2 u_{22} x+u_{23} & u_{21} x^2+u_{22}
%   x+u_{23} & (x+1) u_{23} \\
% x u_{31}+2 x u_{32} & x u_{31}+x u_{32} & x u_{32} & x^2 u_{32} \\
% x u_{41}+2 x u_{42} & x^2 u_{42} & x^2 u_{41} & x u_{41}+u_{42} \\
%\end{array}
%\right).
%$$
It is easy to find $\rk(D_{\mathbb{P}})=4$, thus $\mathbb{P}$ is transformally essential.

By the third step of the  \textbf{Algorithm 1}, the super essential system of $\mathbb{P}$ is $\mathbb{P}_{\TT}=\{\mathbb{P}_0,\mathbb{P}_1,$ $\mathbb{P}_2\}$ with $\TT=\{0,1,2\}$, which is independent of $\mathbb{P}_3$ and $\mathbb{P}_4$.
The symbolic support matrix of $\mathbb{P}_{\TT}$ is
$$D_{\mathbb{P}_{\TT}}=\left(
\begin{array}{cccc}
 2 x u_{01}+2 u_{02} & 2 x u_{01}+u_{02} & x u_{01}+u_{02} & (x+1) u_{02} \\
 2 x u_{11}+2 x u_{12} & 2 x u_{11}+x u_{12} & x u_{11}+x u_{12} & (x^2+x) u_{12} \\
 2 u_{21} x^2+2 u_{22} x+2 u_{23} & 2 u_{21} x^2+2 u_{22} x+u_{23} & u_{21} x^2+u_{22}
   x+u_{23} & (x+1) u_{23} \\
\end{array}
\right).$$
%\begin{figure*} [htp]
% \begin{center} %\renewcommand{\arraystretch}{1.2}
%\tabcolsep4pt{\sc
%$D_{\mathbb{P}_{\TT}}=\left(
%\begin{array}{cccc}
% 2 x u_{01}+2 u_{02} & 2 x u_{01}+u_{02} & x u_{01}+u_{02} & (x+1) u_{02} \\
% 2 x u_{11}+2 x u_{12} & 2 x u_{11}+x u_{12} & x u_{11}+x u_{12} & (x^2+x) u_{12} \\
% 2 u_{21} x^2+2 u_{22} x+2 u_{23} & 2 u_{21} x^2+2 u_{22} x+u_{23} & u_{21} x^2+u_{22}
%   x+u_{23} & (x+1) u_{23} \\
%\end{array}
%\right).$}
%\end{center}
%\vskip-10pt
%\end{figure*}

Since the submatix $M$ of $D_{\mathbb{P}_{\TT}}$ by deleting the middle two columns is
$$A_{14}=\left(
\begin{array}{cc}
 2 x u_{01}+2 u_{02} & (x+1) u_{02} \\
 2 x u_{11}+2 x u_{12} & (x^2+x) u_{12} \\
 2 u_{21} x^2+2 u_{22} x+2 u_{23} & (x+1) u_{23} \\
\end{array}
\right),$$
whose rank is 2, then by Theorem \ref{sr=wsr}, we set $y_2$ and $y_3$ and their transformations to 1, then $\widetilde{\mathbb{P}_{\TT}}=\{\widetilde{\mathbb{P}_0},\widetilde{\mathbb{P}_1},\widetilde{\mathbb{P}_2}\}$, where
%\begin{eqnarray}
%&&\no \widehat{\mathbb{P}}_0=u_{00}[i]+u_{01}[i] \,y[1,i+1]^2 \,y[2,i+1]^2 +u_{02}[i] \,y[2,i] \,y[1,i]^2,\\
%&&\no  \widehat{\mathbb{P}}_1=u_{10}[i]+u_{11}[i] \,y[1,i+1]^2\,y[2,i+1]^2+u_{12}[i] \,y[1,i+1]^2\,y[2,i+1],\\
%&&\no  \widehat{\mathbb{P}}_2=u_{20}[i]+u_{21}[i] \,y[1,i+2]^2 \,y[2,i+2]^2 \\
%   &&\no\hspace{2cm}+u_{22}[i] \,y[1,i+1]^2
%   \,y[2,i+1]^2 +u_{23}[i]\,y[1,i]^2 \,y[2,i].
%\end{eqnarray}
\begin{eqnarray}
&&\no \widetilde{\mathbb{P}}_0=u_{00}+u_{01} \,y_{11}^2 +u_{02} \,y_{1}^2\, y_4\,y_{41},\\
&&\no  \widetilde{\mathbb{P}}_1=u_{10}+u_{11} \,y_{11}^2+u_{12} \,y_{11}^2 \,y_{41}\,y_{42},\\
&&\no  \widetilde{\mathbb{P}}_2=u_{20}+u_{21} \,y_{12}^2 +u_{22}\,y_{11}^2
  +u_{23} \,y_{1}^2\, y_4\,y_{41}.
\end{eqnarray}

Note that by Theorem \ref{sr=wsr}, one can delete any two columns to find the submatrix with rank 2. For example, by deleting the last two columns of $D_{\mathbb{P}_{\TT}}$, one also obtain the submatrix $A_{12}$. The rank of $A_{12}$ is 2 and one can set $y_3$ and $y_4$ and their differences to 1 to get the new simplified super essential difference system.

The order matrix of $\widetilde{\mathbb{P}_{\TT}}$ is
$\left(
\begin{array}{cc}
 1 & 1 \\
 1 & 2 \\
 2 & 1 \\
\end{array}
\right)$,
then the Jacobi numbers are $J_{\hat{0}}=4,J_{\hat{1}}=J_{\hat{2}}=3$. Since the last column   of $A_{14}$ has a common factor $(x+1)$, thus by Proposition \ref{prop-jacobi}, the modified Jacobi numbers are $\widetilde{J}_{\hat{0}}=J_{\hat{0}}-1=3,\widetilde{J}_{\hat{1}}=J_{\hat{1}}-1=2,\widetilde{J}_{\hat{2}}=J_{\hat{2}}-1=2$. Then we use the modified Jacobi numbers $\widetilde{J}_{\hat{i}}\,(i=0,1,2)$ to construct an { algebraic} system $\widehat{\mathbb{P}}=\{\widetilde{\mathbb{P}}_0^{[\widetilde{J}_{\hat{0}}]},
\widetilde{\mathbb{P}}_1^{[\widetilde{J}_{\hat{1}}]},\widetilde{\mathbb{P}}_2^{[\widetilde{J}_{\hat{2}}]}\}$.

The following steps are performed in algebraic circumstance. For the system $\widehat{\mathbb{P}}$, we first search for an essential system with minimal ranking, and then perform a variable transformation for the essential system, i.e.,
\begin{eqnarray}\label{exmp-1}
&&\no\hspace{-0.65cm} \mathcal{\widehat{P}}=\{\widetilde{\mathbb{P}}_0,\sigma\widetilde{\mathbb{P}}_0,\sigma^2\widetilde{\mathbb{P}}_0,\widetilde{\mathbb{P}}_1,
\sigma\widetilde{\mathbb{P}}_1, \widetilde{\mathbb{P}}_2,\sigma\widetilde{\mathbb{P}}_2\}\\
&&\no \hspace{-0.3cm}=\big\{u_{00}+z_4 u_{01}+z_1 u_{02},\sigma u_{00}+z_5 \sigma u_{01}+z_2  \sigma u_{02},\sigma^2u_{00}+z_6 \sigma^2u_{01}+z_3 \sigma^2u_{02},\\
  &&\no \hspace{0.3cm}u_{10}+z_4 u_{11}+z_2 u_{12},\sigma u_{10}+z_5 \sigma u_{11}+z_3 \sigma u_{12},\\
   &&\hspace{0.3cm}u_{20}+z_5 u_{21}+z_4 u_{22}+z_1
   u_{23},\sigma u_{20}+z_6 \sigma u_{21}+z_5 \sigma u_{22}+z_2
    \sigma u_{23}\big\},
\end{eqnarray}
where {$z_1=y_{1}^2 y_{41},z_2=y_{11}^2 y_{42}y_{41},z_3=y_{12}^2 y_{42}, z_4=y_{11}^2, z_5=y_{12}^2, z_6=y_{13}^2$}.

Finally, regard
$\mathcal{\widehat{P}}$ as the algebraic polynomial system of
$z_1,\dots,z_6$ and use the mixed subdivision algorithm in \citep{emiris} to obtain the sparse algebraic resultant of $\widehat{\mathbb{P}}$ which is the required sparse difference resultant $R$. This step cannot be done by hand and will be performed on a computer in next subsection.

 \begin{remark}
The modified Jacobi numbers reduce the dimensions of symbolic support matrices from $13\times12$ to $10\times 9$ in searching the algebraic essential system with minimal ranking. Moreover, the modified Jacobi numbers will drop the complexity of the searching algorithm for sparse difference resultant in \citep{gao-2015}.
\end{remark}

\textbf{2. Implementation}

We use the complied package \textbf{SDResultant} and the mixed subdivision algorithm to automatically compute the sparse difference resultant of $\mathbb{P} $. The program can be found in

\noindent https://github.com/cmyuanmmrc/codeforsdr.

{\color{black} Note that, recently, a package to compute the sparse resultant of algebraic polynomial system is given in Macaulay2 in \cite{Staglian}. Thus maybe it is an alternative way to look for the sparse resultant of the algebraic case in our algorithm.}

Firstly, we transform the target difference polynomial system $\mathbb{P} $ to the given form.
Input the difference polynomial system $\mathbb{P} =
\{\mathbb{P}_0,\mathbb{P}_1,\mathbb{P}_2,\mathbb{P}_3,\mathbb{P}_4\}$ in the form
\begin{eqnarray}\label{xemple1}
&&\no \mathbb{P}_0=u{00}(i)+u{01}(i) \,y(1,i+1)^2 \,y(2,i+1)^2 \,y(3,i+1)\\
   &&\no\hspace{1cm}+u{02}(i) \,y(1,i)^2 \,y(2,i) \,y(3,i) \,y(4,i)\,y(4,i+1),\\
&&\no  \mathbb{P}_1=u{10}(i)+u{11}(i) \,y(1,i+1)^2 \,y(2,i+1)^2 \,y(3,i+1)
\\
   &&\hspace{1cm}+u{12}(i)\,y(1,i+1)^2 \,y(2,i+1) \,y(3,i+1) \,y(4,i+1)\,y(4,i+2),\\
&&\no  \mathbb{P}_2=u{20}(i)+u{21}(i) \,y(1,i+2)^2 \,y(2,i+2)^2 \,y(3,i+2)\\
   &&\no\hspace{1cm}+u{22}(i) \,y(1,i+1)^2
   \,y(2,i+1)^2 \,y(3,i+1)+u{23}(i) \,y(1,i)^2\,y(2,i) \,y(3,i) \,y(4,i)\,y(4,i+1),\\
&&\no  \mathbb{P}_3=u{30}(i)+u{31}(i) \,y(1,i+1)\,y(2,i+1)+u{32}(i)\,y(1,i+1)^2 \,y(2,i+1) \,y(3,i+1)
   \,y(4,i+2),\\
&&\no  \mathbb{P}_4=u{40}(i)+u{41}(i) \,y(1,i+1)\,y(3,i+2) \,y(4,i+1)+u{42}(i) \,y(1,i+1)^2 \,y(2,i+2) \,y(4,i).
\end{eqnarray}

{\color{black}Firstly, applying the package \textbf{SDResultant} to system (\ref{exmp-1}) gives a strong essential system
\begin{eqnarray}\label{exm-1o}
&&\no\big\{\,\text{u00}(i)+z(5) z(3)^2 \text{u01}(i)+z(1) \text{u02}(i),\text{u00}(i+1)+z(4)^2 z(6) \text{u01}(i+1)+z(3) \text{u02}(i+1),\\
&&\no\text{u00}(i+2)+z(2) \text{u01}(i+2)+z(4) \text{u02}(i+2),\text{u10}(i)+z(5) z(3)^2 \text{u11}(i)+z(3) \text{u12}(i),\\
&&\no\text{u10}(i+1)+z(6) z(4)^2 \text{u11}(i+1)+z(4) \text{u12}(i+1),\\
&&\no\text{u20}(i)+z(4)^2 z(6) \text{u21}(i)+z(5) z(3)^2 \text{u22}(i)+z(1) \text{u23}(i),\\
&&\no\text{u20}(i+1)+z(2) \text{u21}(i+1)+z(6) z(4)^2 \text{u22}(i+1)+z(3) \text{u23}(i+1)~\big\},
\end{eqnarray}
where $z(i)$ are the same as $z_i$ in system (\ref{exmp-1}). It takes $50.796$ seconds by the order \emph{TimeUsed[]} in Mathematica 10.
Observe that the first three polynomials $\mathbb{P}_0,\mathbb{P}_1$ and $\mathbb{P}_2$  constitute a super-essential system while the last two $\mathbb{P}_3$ and $\mathbb{P}_4$ are redundant. Obviously, system (\ref{exm-1}) can be further simplified to the following form
\begin{eqnarray}\label{exm-1}
&&\no\big\{~{u00}(i)+w(4) {u01}(i)+w(1) {u02}(i),{u00}(i+1)+w(6) {u01}(i+1)+w(3) {u02}(i+1),\\
&&\no {u00}(i+2)+w(2)
   {u01}(i+2)+w(5) {u02}(i+2),{u10}(i)+w(4) {u11}(i)+w(3) {u12}(i),\\
&&\no{u10}(i+1)+w(6) {u11}(i+1)+w(5){u12}(i+1),\\
&&\no{u20}(i)+w(6) {u21}(i)+w(4) {u22}(i)+w(1) {u23}(i),\\
&&{u20}(i+1)+w(2) {u21}(i+1)+w(6){u22}(i+1)+w(3) {u23}(i+1)~\big\},
\end{eqnarray}
where $w(i)=z(i)\,(i=1,2,3,5),w(4)=z(5)z(3)^2,w(6)=z(4)^2z(6) $.

Then we regard system (\ref{exm-1}) as an algebraic polynomial system in $w(i)\,(i=1,\dots,6)$, and then with the mixed subdivision algorithm find the matrix representation of the sparse resultant  $R$ in the form
\begin{eqnarray}
M=\left(
\begin{array}{ccccccc}
 \text{u10}(i+1) & \text{u11}(i+1) & \text{u12}(i+1) & 0 & 0 & 0 & 0 \\
 \text{u20}(i) & \text{u21}(i) & 0 & \text{u22}(i) & 0 & 0 & \text{u23}(i) \\
 \text{u00}(i+2) & 0 & \text{u02}(i+2) & 0 & 0 & \text{u01}(i+2) & 0 \\
 \text{u10}(i) & 0 & 0 & \text{u11}(i) & \text{u12}(i) & 0 & 0 \\
 \text{u00}(i+1) & \text{u01}(i+1) & 0 & 0 & \text{u02}(i+1) & 0 & 0 \\
 \text{u20}(i+1) & \text{u22}(i+1) & 0 & 0 & \text{u23}(i+1) & \text{u21}(i+1) & 0 \\
 \text{u00}(i) & 0 & 0 & \text{u01}(i) & 0 & 0 & \text{u02}(i) \\
\end{array}
\right)
\end{eqnarray}

The time to construct the matrix in Maple 18 is 4.641 seconds.
Then the required sparse difference resultant $R$ of system (\ref{xemple1}) is
  \begin{eqnarray}
  &&\no\hspace{-0.6cm} R=Det[M]\\
&&\no\hspace{-0.3cm}={u01}(2+i) {u02}(i) {u02}(1+i) {u11}(i) {u12}(1+i) {u20}(1+i) {u21}(i)\\
&&\no-{u02}(i) {u02}(1+i) {u02}(2+i)
   {u11}(i) {u11}(1+i) {u20}(i) {u21}(1+i)\\
&&\no+{u02}(i) {u02}(1+i) {u02}(2+i) {u10}(1+i) {u11}(i) {u21}(i)
   {u21}(1+i)\\
&&\no-{u00}(2+i) {u02}(i) {u02}(1+i) {u11}(i) {u12}(1+i) {u21}(i) {u21}(1+i)\\
&&\no+{u01}(1+i)
   {u01}(2+i) {u02}(i) {u12}(i) {u12}(1+i) {u20}(1+i) {u22}(i)\\
&&\no+{u02}(i) {u02}(1+i) {u02}(2+i) {u10}(i)
   {u11}(1+i) {u21}(1+i) {u22}(i)\\
&&\no+{u01}(1+i) {u02}(i) {u02}(2+i) {u10}(1+i) {u12}(i) {u21}(1+i)
   {u22}(i)\\
&&\no-{u00}(1+i) {u02}(i) {u02}(2+i) {u11}(1+i) {u12}(i) {u21}(1+i) {u22}(i)\\
&&\no-{u00}(2+i)
   {u01}(1+i) {u02}(i) {u12}(i) {u12}(1+i) {u21}(1+i) {u22}(i)\\
&&\no-{u01}(2+i) {u02}(i) {u02}(1+i) {u11}(i)
   {u12}(1+i) {u20}(i) {u22}(1+i)\\
&&\no+{u01}(2+i) {u02}(i) {u02}(1+i) {u10}(i) {u12}(1+i) {u22}(i)
   {u22}(1+i)\\
&&\no-{u00}(1+i) {u01}(2+i) {u02}(i) {u12}(i) {u12}(1+i) {u22}(i) {u22}(1+i)\\
&&\no-{u01}(i)
   {u01}(1+i) {u01}(2+i) {u12}(i) {u12}(1+i) {u20}(1+i) {u23}(i)\\
&&\no-{u01}(i) {u02}(1+i) {u02}(2+i)
   {u10}(i) {u11}(1+i) {u21}(1+i) {u23}(i)\\
&&\no+{u00}(i) {u02}(1+i) {u02}(2+i) {u11}(i) {u11}(1+i)
   {u21}(1+i) {u23}(i)\\
&&\no-{u01}(i) {u01}(1+i) {u02}(2+i) {u10}(1+i) {u12}(i) {u21}(1+i)
   {u23}(i)\\
&&\no+{u00}(1+i) {u01}(i) {u02}(2+i) {u11}(1+i) {u12}(i) {u21}(1+i) {u23}(i)
\\
&&\no+{u00}(2+i) {u01}(i)
   {u01}(1+i) {u12}(i) {u12}(1+i) {u21}(1+i) {u23}(i)\\
&&\no-{u01}(i) {u01}(2+i) {u02}(1+i) {u10}(i)
   {u12}(1+i) {u22}(1+i) {u23}(i)\\
&&\no+{u00}(i) {u01}(2+i) {u02}(1+i) {u11}(i) {u12}(1+i) {u22}(1+i)
   {u23}(i) \\
&&\no+{u00}(1+i) {u01}(i) {u01}(2+i) {u12}(i) {u12}(1+i) {u22}(1+i) {u23}(i)\\
&&\no+{u01}(1+i)
   {u01}(2+i) {u02}(i) {u11}(i) {u12}(1+i) {u20}(i) {u23}(1+i)\\
&&\no-{u00}(1+i) {u01}(2+i) {u02}(i) {u11}(i)
   {u12}(1+i) {u21}(i) {u23}(1+i)\\
&&\no-{u01}(1+i) {u01}(2+i) {u02}(i) {u10}(i) {u12}(1+i) {u22}(i)
   {u23}(1+i)\\
&&\no+{u01}(i) {u01}(1+i) {u01}(2+i) {u10}(i) {u12}(1+i) {u23}(i) {u23}(1+i)\\
&&\no-{u00}(i)
   {u01}(1+i) {u01}(2+i) {u11}(i) {u12}(1+i) {u23}(i) {u23}(1+i).
     \end{eqnarray}

Thus the total time to compute the sparse difference resultant $R$ of system (\ref{xemple1}) is $51.688+4.641=56.329$ seconds. Though the time is not very good, it is the first computable algorithm for the sparse difference resultant while the algorithm in \cite{gao-2015} is not performed on computer.  }    %{\sc$$

\subsubsection{Several practical examples}
We show how to solve difference problems with the sparse difference resultant algorithm by considering several practical examples.

\textbf{Example 1. The $n$-th Fibonacci number}\vspace{0.1cm}

{ The second example is also about the $n$-th Fibonacci number $F_n$. We show the sequence $A_n:=F_{2^n}$ satisfy a nonlinear difference equation \citep{ekhad-2014}. Let $B_n=F_{2^n+1}$. Then standard identities of Fibonacci numbers implies two difference equations \citep{ovc-2018}
\begin{eqnarray}\label{fib-l}
&& P_0=A_{n+1}-A_n(2B_n-A_n)=0,~~~ P_1=B_{n+1}-A_n^2-B_n^2=0.
\end{eqnarray}

After rewriting system (\ref{fib-l}) in the Mathematica form $P_0=u(i+1)-u(i)(2y(1,i)-u(i)), P_1=y(1,i+1)-u(i)^2-y(1,i)^2$, then with our \textbf{SDResultant} algorithm, we take 0.922 seconds to obtain the strong essential system
$$\big\{-2 z(1) u(i)+u(i)^2+u(i+1),-2 z(2) u(i+1)+u(i+1)^2+u(i+2),-u(i)^2-z(1)^2+z(2)\big\}. $$
Then eliminate the variables $z(1)$ and $z(2)$ by means of the mixed subdivision algorithm, we take $0.015$ seconds and obtain a $4\times4$ matrix whose determinant is $-5 u(i+1) u(i)^4+2 u(i+2) u(i)^2-u(i+1)^3=0$ which is the required nonlinear difference equation in $A_n$ by recovering the original variables.\vspace{0.1cm}

\textbf{Example 2. The stage structured Leslie-Gower model}\vspace{0.1cm}

The third example is to consider the stage structured Leslie-Gower model \citep{hensen-2007}
\begin{eqnarray}\label{les-l}
&&\no P_0=(1+d_1A_n)J_{n+1}-b_1A_n,\\
&&\no P_1=(1+J_n+c_1j_n)A_{n+1}-e_1J_n, \\
&&\no P_2=(1+d_2a_n)j_{n+1}-b_2a_n,\\
&&P_3=(1+c_2J_n+j_n)a_{n+1}-e_2j_n.
\end{eqnarray}

We want to eliminate $A_n,J_n$ and $j_n$ to find the relation of $a_n$. Again rewriting system (\ref{les-l}) in the Mathematica form, and by means of the package \textbf{SDResultant}, we take 40.499 seconds to find the strong essential system
\begin{eqnarray}\label{lesg}
&&\no \big\{-b_1 z(1)+d_1 z(3) z(1)+z(3),c_1 z(4) z(1)-e_1 z(2)+z(2) z(1)+z(1),\\
   &&\no\hspace{0.3cm} z(5) \left(d_2 u(i+2)+1\right)-b_2 u(i+2),c_2 z(2) u(i+2)+z(4)
   \left(u(i+2)-e_2\right)+u(i+2),\\
   &&\hspace{0.3cm}z(4) \left(d_2
   u(i)+1\right)-b_2 u(i),c_2 z(3) u(i+4)+z(5) \left(u(i+4)-e_2\right)+u(i+4)\big\},
\end{eqnarray}
where $u[i]=a_n$ and $z(i)\,(i=1,\dots,5)$ are temporary variables.

Then one can use mixed subdivision method to get the sparse resultant of system (\ref{lesg}),
 we spend 0.954 seconds to obtain a $8\times8$ matrix $M_1$ and a $2\times2$ matrix $M_2$. Then the quotient of $|M_1|/|M_2|$ is the condition for $a_n$
\begin{eqnarray}
&&\no a_{n+4} \left(d_2 a_{n+2}+1\right) \left(\Delta-b_1 c_2 e_1 a_{n+2} \left(d_2
   a_{n}+1\right)\right)\\
   &&\no +b_2 \Big[a_{n+2} a_{n+4} \left(\Delta-b_1 c_2 e_1 a_{n} \left(d_2
   a_{n+2}+1\right)\right)+e_2 a_{n+2} \left(b_1 c_2 d_2 e_1 a_{n}
   a_{n+4}-\Delta\right)+e_2 b_1 c_2 e_1 a_{n} a_{n+4}\Big],
\end{eqnarray}
where
$\Delta=b_2 e_2 \left(d_1 e_1+1\right) a_n-a_{n+2} \left(d_1 e_1+1\right)
   \left(\left(b_2+d_2\right) a_n+1\right)-c_2 a_{n+2}\left(a_n \left(b_2
   c_1+d_2\right)+1\right)$.\vspace{0.1cm}

\textbf{Example 3. The May-Leonard model for 2-plant annual competition}\vspace{0.1cm}

The fourth example is to consider the May-Leonard model for 2-plant annual competition.  We first verify whether $y_n$ can be eliminated from the May-Leonard model for 2-plant annual competition which was considered in \citep{roeger-2004,ovc-2018},
\begin{eqnarray}\label{may-2}
&&\no P_0=(x_{n+1}-bx_n)(x_n+\alpha_1 y_n)+(b-1)x_n,\\
&& P_1=(y_{n+1}-by_n)(\alpha_2 x_n+ y_n)+(b-1)y_n,
\end{eqnarray}
where $b,\,\alpha_i,\,(i=1,2)$ are parameters.

We regard $x_n$ as a parameter variable and $y_n$ as unknown variable, then
the Jacobi numbers are $J_0=1, J_1=0$.
Then difference $P_0$ by once, and together with $P_1$
we obtain three difference polynomials. By the package \textbf{SDResultant}, input system (\ref{may-2}) in Mathematica form
\begin{eqnarray}
&&\no P_0=(u(i+1)-b u(i)) \left(\alpha _1 y(1,i)+u(i)\right)+(b-1) u(i),\\
&&\no P_1=(y(1,i+1)-b y(1,i)) \left(+y(1,i)+\alpha_2 u(i)\right)+(b-1) y(1,i),
\end{eqnarray}
where $x_n=u(i), y_n=y(1,i)$. Then we spend 2.418 seconds to find the strong essential system
\begin{eqnarray}\label{mayr-2d}
&&\no\big\{z(1) \left(\alpha _1 u(i+1)-\alpha _1 b u(i)\right)+u(i) (-bu(i)+b+u(i+1)-1),\\
&&\no \hspace{0.2cm}z(2)\left(\alpha _1 u(i+2)-\alpha _1 b u(i+1)\right)+u(i+1) ( -bu(i+1)+b+u(i+2)-1),\\
&& \hspace{0.2cm}z(1)\left(-\alpha _2 b u(i)+b-1\right)-b z(1)^2+\alpha _2 z(2) u(i)+z(2) z(1)\big\}.
\end{eqnarray}

Then with the mixed subdivision method to eliminate the variables $z(1)$ and $z(2)$ from system (\ref{mayr-2d}), we get a $4\times4$ square matrix whose determinant is the the difference polynomial after recovering $u[i]$ to $x_n$
\begin{eqnarray}
&&\no\left(\alpha _1 \alpha _2-1\right) b^3 x_n^3 \left(b x_{n+1}-x_{n+2}\right)+b^2 x_n^2 \Big[x_{n+1} \left((b-1) \left(-\alpha _1
   \left(\alpha _2 (b-1)+b\right)+2 b-1\right)\right)
   x_{n+2}\\
   &&\no +3x_{n+1}
   \left(\alpha _1 \alpha _2-1\right) (1-b x_{n+1})+\left(\alpha _1 \left(\alpha _2+1\right)-2\right) (b-1) x_{n+2}\Big]+\alpha _1 \alpha _2 x_{n+1}^3
  (x_{n+2}-b x_{n+1}+b-1)\\
   &&\no +b x_n
   \Big[3 \left(\alpha _1 \alpha _2-1\right)x_{n+1}^3(bx_{n+1}-x_{n+2})+x_{n+1}^2 \left((b-1) \left(\alpha _1
\alpha _2 (b-2)+2 b\alpha _1-3 b+2\right)\right)\\
   &&\no+(b-1) x_{n+1} \left((b-1) \left(\alpha _1 b-b+1\right)-\left(\alpha _1
   \left(\alpha _2+2\right)-3\right) x_{n+2}\right)-\left(\alpha _1-1\right) (b-1)^2
   x_{n+2}\Big]\\
   &&\no +x_{n+1} \left(b+x_{n+1}-1\right) \left(-x_{n+1} \left(\alpha _1 (b-1)
   b+b+x_{n+2}-1\right)+\alpha _1 (b-1) x_{n+2}+b x_{n+1}^2\right).
\end{eqnarray}

\textbf{Example 4. The May-Leonard model for 3-plant annual competition}\vspace{0.1cm}

Finally, we consider the May-Leonard model for 3-plant annual competition  and verify whether $y_n$ and $z_n$ can be eliminated from it which was considered in \citep{roeger-2004,ovc-2018},}
\begin{eqnarray}\label{may-l}
&&\no P_0=(x_{n+1}-bx_n)(x_n+\alpha_1 y_n+\beta_1 z_n)+(b-1)x_n,\\
&&\no P_1=(y_{n+1}-by_n)(\alpha_2 x_n+ y_n+\beta_2 z_n)+(b-1)y_n, \\
&& P_2=(z_{n+1}-bz_n)(\alpha_3 x_n+\beta_3 y_n+ z_n)+(b-1)z_n,
\end{eqnarray}
where $b,\,\alpha_i,\,\beta_i\,(i=1,2,3)$ are parameters.
 By means of \textbf{SDResultant} algorithm, we take $x_n$ as a parameter variable and $y_n,z_n$ as two unknown variables, then
the Jacobi numbers are $J_0= 2, J_1=J_2=1$.
Then difference $P_0$ by twice, $P_1$ and $P_2$ by once respectively, and
we obtain seven difference polynomials. Again by the package \textbf{SDResultant}, input system (\ref{may-l}) in Mathematica form
\begin{eqnarray}
&&\no P_0=(u(i+1)-b u(i)) \left(\alpha _1 y(1,i)+\beta _1 y(2,i)+u(i)\right)+(b-1) u(i),\\
&&\no P_1=(y(1,i+1)-b y(1,i)) \left(\beta _2 y(2,i)+y(1,i)+\alpha
   _2 u(i)\right)+(b-1) y(1,i),\\
&&\no P_2=(y(2,i+1)-b y(2,i)) \left(\beta _3 y(1,i)+y(2,i)+\alpha _3 u(i)\right)+(b-1) y(2,i),
\end{eqnarray}
where $x_n=u(i), y_n=y(1,i),z_n=y(2,i)$, it takes 12.032 seconds to obtain the strong essential system
\begin{eqnarray}\label{diff-m}
&&\no\hspace{-0.2cm}\big\{(u(i+1)-b u(i)) \left(u(i)+\alpha _1 z(1)+\beta _1 z(4)\right)+(b-1) u(i),\\
&&\no(u(i+2)-b u(i+1)) \left(u(i+1)+\alpha _1 z(2)+\beta _1
   z(5)\right)+(b-1) u(i+1),\\
&&\no(u(i+3)-b u(i+2)) \left(u(i+2)+\alpha _1 z(3)+\beta _1 z(6)\right)+(b-1) u(i+2),\\
&&\no(z(2)-b z(1)) \left(\alpha _2
   u(i)+\beta _2 z(4)+z(1)\right)+(b-1) z(1),\\
&&\no(z(3)-b z(2)) \left(\alpha _2 u(i+1)+\beta _2 z(5)+z(2)\right)+(b-1) z(2),\\
&&\no(z(5)-b z(4))
   \left(\alpha _3 u(i)+\beta _3 z(1)+z(4)\right)+(b-1) z(4),\\
&&(z(6)-b z(5)) \left(\alpha _3 u(i+1)+\beta _3 z(2)+z(5)\right)+(b-1)
   z(5)\big\}.
   \end{eqnarray}

Then regarding (\ref{diff-m}) as an algebraic polynomial system about $z(i)\, (i=1,\dots,6)$, with the mixed subdivision method we spend $16.812$ seconds to get the sparse resultant of system (\ref{diff-m}) which is the sparse difference resultant of system (\ref{may-l}). The sparse difference resultant is the quotient of two determinants, the determinant of a $163\times163$ matrix divided by the one of a $122\times122$ matrix, of total degree at most $41$ in the parameter variables. {\color{black}In general, the computation of determinant of large matrixes including parameters is not an easy work while for some specializations of the parameters  one can resort to division method and the greatest common divisor technique \citep{canny-2000}.} Note that the sparse difference resultant is not equal to zero identically since it is not zero for the particular values $b=2,\alpha_3=\beta_1=0$ and the other parameters $\alpha _1=\alpha _2=\beta _2=\beta _3=1$, but it is still very large.

In Table 1, the runtimes of computing the sparse difference resultants for the above examples are collected in the first three columns while the runtimes of the algorithm in \cite{ovc-2018} are listed in the fourth column.
\begin{table}[htb]
 \centering
\caption{Comparisons of the runtimes and the sizes of matrix}
\footnotesize\renewcommand{\arraystretch}{1.2}
\begin{tabular}{ccccccc}\hline
\multicolumn{1}{c}{Examples}&\multicolumn{1}{c}{Difference(s)}&\multicolumn{1}{c}{Algebraic(s)}&\multicolumn{1}{c}{Total(s)}&\multicolumn{1}{c}{\cite{ovc-2018}}
&\multicolumn{1}{c}{Our sizes}&\multicolumn{1}{c}{\cite{gao-2015}}\\\hline
Artificial& 94.614& 4.641 &99.256&\#&$10\times9$&$13\times12$\\
1&0.922&0.015&0.938& 0.06&$3\times2$&$3\times2$\\
2&40.499&0.954&41.456& 18.71 &$13\times12$&$13\times12$\\
3& 2.418&0.063&2.483&0.6 &$3\times2$&$3\times2$\\
4&12.032 &16.812&$>1000$& $>1000$&$7\times6$&$7\times6$ \\
   \hline
 \end{tabular}\\\vspace{0.1cm}
  \end{table}
In particular, the column ``Difference" represent the running time implemented in Mathematica for finding the strong essential polynomial system while the column ``Algebraic" are the ones using the mixed subdivision method in \cite{emiris} to get the matrix representation of sparse difference resultant. {\color{black}The column ``Total" denote the sums of the running time in the two columns ``Difference" and ``Algebraic" and the running time of computing the determinants found by the function ``AbsoluteTiming()" in Mathematica. Note that the running time for computing the determinant of lower-order matrixes are very short, but for Example~4, we can not get the explicit expression from the matrix representation within 1000 seconds.} The last two columns are the comparisons of the matrix sizes of our algorithm with the one in \cite{gao-2015}. The matrixes for the comparison are the symbolic support matrixes, which respectively correspond to the difference polynomial system by difference the new super-essential system in step 4 with the new Jacobi number and the ones by difference the super-essential system in step 3 with the original Jacobi number.

Table 1 shows that for the simple examples, the runtimes by the algorithm in \cite{ovc-2018} is shorter than ours, but for some complex example, such as the artificial example, our algorithm can give the matrix representation of the sparse difference resultant and find the corresponding eliminated polynomial, {\color{black} while with the package in \cite{ovc-2018} we encounter an error. In particular, for example 4, both our algorithm and the one in \cite{ovc-2018} runs the times larger than 1000 seconds, where the difficulty for the former is the computation of determinant of large matrixes while the latter calls Gr\"{o}bner basis algorithm, but our algorithm can represent the eliminated polynomial as the quotient of determinants of two matrixes and give the explicit result for some specializations of the parameters}.
% but the algorithm in \cite{ovc-2018} runs the time larger than 1000 seconds and does not give the final result since they use to find the result
\section{Conclusion}
Sparse difference resultant for a Laurent transformally essential system is further studied and new order bounds are obtained.
We use the difference specialization technique to simplify the computation of the sparse difference resultant.
Based on these results, we give an efficient algorithm to compute sparse difference resultants.
We analyze the complexity of the algorithm and illustrate the efficiency by an artificial and several practical examples. To the best of our knowledge, our algorithm and its implementation are the first automatic algorithm to compute the sparse difference resultant and to give the matrix representations. The polynomial time complexity makes our algorithm more efficient for some complicated examples.
\section*{Acknowledgments}
We sincerely thank for the valuable suggestions of anonymous referees and editors which greatly helped us to improve this work.\\
%This paper is partially supported by the National Natural Science Foundation
%of China (Nos. 11671014 and 11688101) and the Youth Innovation Promotion Association of CAS.
%\section{Main Result}

%
%We recall the works by \cite{Key1}.
%
%We recall the works in \citep{Key1}.
%
%\begin{definition}[Nice]
%We say a polynomial is \emph{nice} if ....  \qed
%\end{definition}
%
%\begin{theorem}[Wonderful]
%All nice polynomials are also wonderful....  \qed
%\end{theorem}
%
%\begin{proof}
%The proof is easy...
%\end{proof}

\bibliographystyle{elsarticle-harv}

\begin{thebibliography}{30}
%\bibitem[{Doe (1996)}]{Key1}
%Doe, Jane, Great result, Superb Journal, 1996
\bibitem[{Cox (2004)}]{cox-2004} D.A. Cox, J. Little, D. O'Shea, Using algebraic Geometry. Springer-Verlag,  New  York, 2004.

\bibitem[{Gelfand (1994)}]{gelfand}
 I.~M. Gelfand,  M. Kapranov, A. Zelevinsky, {\em Discriminants, Resultants and Multidimensional Determinants}.
 Boston, Birkh\"auser, 1994.

\bibitem[{Sturmfels (1993)}]{sturmfels} B. Sturmfels, Sparse Elimination Theory, In {\em Computational Algebraic Geometry and Commutative Algebra}, Eisenbud, D., Robbiano, L. eds. 264-298, Cambridge University Press, 1993.

 \bibitem[{Canny (1995)}]{emiris} J.F.  Canny,  I.Z. Emiris, Efficient Incremental Algorithms for the Sparse Resultant and the
Mixed Volume. Journal of Symbolic Computation 20(2) (1995) 117-149.

{\color{black}\bibitem[{Emiris (2012a)}]{emiris-2012a} I.Z. Emiris, A. Mantzaflaris, Multihomogeneous resultant formulae for systems with scaled support, Journal of Symbolic Computation 47 (7) (2012) 820-842.

\bibitem[{Emiris (2012b)}]{emiris-2012b} I.Z. Emiris, A General Solver Based on Sparse Resultants, arXiv:1201.5810 (2012).}

\bibitem[{D'Andrea (2011)}]{an-2011}C. D'Andrea, Macaulay style formulas for sparse resultants. Transactions of the American Mathematical Society 354(7) (2002) 2595-2629.

\bibitem[{Yang (2011)}]{yang-2011}  L. Yang, Z. Zeng, W. Zhang, Differential elimination with
Dixon resultants. Applied Mathematics and Computation 218 (2011) 10679-10690.

\bibitem[{Rueda (2010)}]{rueda-2010}S. L. Rueda, J. R. Sendra, Linear complete differential resultants and the implicitization of
linear DPPEs. Journal of Symbolic Computation 45 (2010) 324-341.

\bibitem[{Li (2015a)}]{gao-2015-a} W. Li, C.M. Yuan, X.S. Gao, Sparse Differential Resultant for Laurent Differential
Polynomials. Foundation of  Computational Mathematics 15 (2015) 451-517.

\bibitem[{Li (2015b)}]{gao-2015}W. Li, C.M. Yuan, X.S. Gao, Sparse difference resultant.
Journal of Symbolic Computation 68 (2015) 169-203.

\bibitem[{Cohn (1965)}]{cohn} R.M. Cohn, {\em Difference Algebra}. Interscience Publishers.  New York,  1965.


% \bibitem{dandrea1} C. D'Andrea, Macaulay Style Formulas for Sparse Resultants, Transactions of the American Mathematical Society 354(7)(2002) 2595-2629.
\bibitem[{Hrushovski (2007)}]{hru-2007} E. Hrushovski, F. Point, On von Neumann regular rings with an automorphism. Journal of
Algebra 315(1) (2007) 76-120.

\bibitem[{Ovchinnikov (2020)}]{ovc-2018}A. Ovchinnikov, G. Pogudin, T. Scanlon, Effective difference elimination and Nullstellensatz.
Journal of the European Mathematical Society 22(8) (2020) 2419-2452.

\bibitem[{Gao (2009)}]{gao-2009} X.S. Gao, Y. Luo, C.M. Yuan, A characteristic set method for ordinary difference polynomial systems. Journal of Symbolic Computation 44 (2009) 242-260.

\bibitem[{Moenck (1973)}]{moenck-1973}R.T. Moenck, {\em Fast computation of GCDs}. Proc. STOC'73, ACM Press, (1973) 142-151.

 \bibitem[{Ollivier (2010)}]{jacobi-number} F. Ollivier,  Jacobi's bound and normal forms computations. A historical survey,  arXiv:0911.2674v2 (2010).

 \bibitem[{Sturmfels (1994)}]{sturmfels2} B. Sturmfels, On The Newton Polytope of the Resultant. Journal of Algebraic Combinatorics 3 (1994) 207-236.


\bibitem[{Canny (2000)}]{canny-2000}J.F. Canny, I.Z. Emiris, A Subdivision-Based Algorithm for the Sparse Resultant. Journal of the ACM 47(3)(2000) 417-451.

{
\bibitem[{Zippel (1979)}]{zippel} R. Zippel, Probabilistic algorithms for sparse polynomials. Lecture Notes in Computer Science 72 (1979) 216-226.

\bibitem[Storjohann (2000)]{Storjohann2000} A.~Storjohann,
 \textit{Algorithms for matrix canonical forms[D]},
 PhD Thesis, Swiss Federal Institute of Technology
  2000.

\bibitem[{Henson (2007)}]{hensen-2007} J. Cushing, S. Henson, L. Roeger, Coexistence of competing juvenile-adult structured populations. Journal of Biological Dynamics 1(2)(2007) 201-231 .

\bibitem[{Roeger (2004)}]{roeger-2004} L. Roeger, L. Allen, Discrete May-Leonard competition models I. Journal of Difference Equations and Applications 10(1) (2004) 77-98.

\bibitem[{Ekhad (2014)}]{ekhad-2014}S.B.  Ekhad, D. Zeilberger, How to generate as many Somos-like miracles as you
wish. Journal of Difference Equations and Applications 20 (2014) 852-858.}


{\bibitem[{Recio (2018)}]{recio-2018}
 T. Recio, J.R. Sendra,  C.Villarino, The importance of being zero. Proc. ISSAC'2018 (2018) 327-333.}

{\color{black}\bibitem[{Staglian (2020)}]{Staglian} G. Staglian$\grave{\text{ o}}$, A package for computations with sparse resultants, 2020. URL https://arxiv.
org/abs/2010.00286.}
\end{thebibliography}
% Include the ".bib" file (generated by bibtex) right here.

\end{document}